\documentclass[11pt]{article}

\usepackage{times}
\usepackage{amsmath,amssymb,amsthm,amsfonts,latexsym,bbm,xspace,graphicx,float,mathtools,dsfont}
\usepackage[backref,colorlinks,citecolor=blue,bookmarks=true]{hyperref}
\usepackage{color}
\usepackage{comment}
\usepackage{algorithm}
\usepackage{algorithmic}
\newcommand{\ind}{\mathds{1}}
\usepackage{tweaklist}
\usepackage{fullpage}
\usepackage{appendix}
\usepackage{cleveref}
\usepackage{accents}

\newtheorem{theorem}{Theorem}

\newtheorem{corollary}{Corollary}
\newtheorem{lemma}{Lemma}

\newtheorem{definition}{Definition}

\newtheorem{informaltheorem}{Informal Theorem}

\newcommand{\opt}{\textsc{OPT}}
\newcommand{\som}{\textsc{GFT}_\textsc{SOM}}
\newcommand{\bom}{\textsc{GFT}_\textsc{BOM}}
\newcommand{\gsom}{\textsc{GFT}_\textsc{GSOM}}
\newcommand{\gbom}{\textsc{GFT}_\textsc{GBOM}}
\newcommand{\gft}{\textsc{GFT}}

\newenvironment{prevproof}[2]{\noindent {\em {Proof of {#1}~\ref{#2}:}}}{$\Box$\vskip \belowdisplayskip}

\newcommand{\Bs}{\bold{s}}
\newcommand{\Bb}{\bold{b}}

\newcommand{\johannesnote}[1]{{\color{green}{#1}}}
\newcommand{\notshow}[1]{{}}

\def \MM  {\mathcal{M}}

\def \L{{\mathcal{L}}}

\DeclareMathOperator{\argmax}{argmax}

\definecolor{MyGray}{rgb}{0.8,0.8,0.8}

\title{Approximating Gains from Trade in Two-sided Markets \\via Simple Mechanisms}

\author{Johannes Brustle\footnote{School of Computer Science, McGill University. \href{mailto:johannes.brustle@mail.mcgill.ca}{johannes.brustle@mail.mcgill.ca}.}\and
Yang Cai\footnote{School of Computer Science, McGill University. \href{mailto:cai@cs.mcgill.ca}{cai@cs.mcgill.ca}.}\and
Fa Wu\footnote{Zhejiang University and McGill University. \href{mailto:wufa85@zju.edu.cn}{wufa85@zju.edu.cn}. Work done in part while the author was visiting McGill University.}\and
Mingfei Zhao\footnote{School of Computer Science, McGill University. \href{mailto:mingfei.zhao@mail.mcgill.ca}{mingfei.zhao@mail.mcgill.ca}.}}
\addtocounter{page}{-1}

\begin{document}
\maketitle
\begin{abstract}
We design simple mechanisms to approximate the Gains from Trade (GFT) in two-sided markets with multiple unit-supply sellers and multiple unit-demand buyers. A classical impossibility result by Myerson and Satterthwaite~\cite{MyersonS83} showed that even with only one seller and one buyer, no Individually Rational (IR), Bayesian Incentive Compatible (BIC) and Budget-Balanced (BB) mechanism can achieve full GFT (trade whenever buyer's value is higher than the seller's cost). On the other hand, they proposed the ``second-best'' mechanism that maximizes the GFT subject to IR, BIC and BB constraints, which is unfortunately rather complex for even the single-seller single-buyer case. Our mechanism is simple, IR, BIC and BB, and achieves $\frac{1}{2}$ of the optimal GFT among all IR, BIC and BB mechanisms. Our result holds for arbitrary distributions of the buyers' and sellers' values and can accommodate any downward-closed feasibility constraints over the allocations. The analysis of our mechanism is facilitated by extending the Cai-Weinberg-Devanur duality framework~\cite{CaiDW16} to two-sided markets.
\end{abstract}
\thispagestyle{empty}

\newpage
\section{Introduction}
Mechanism Design for one-sided market has been extensively studied in Economics in the past few decades and recently investigated in Computer Science. In a one-sided market, the mechanism aims to allocate the items to agents and charge them payments so that (1) the agents are {incentivized} to reveal their true preferences and (2) a certain objective is (approximately) optimized, e.g., social welfare and revenue. Various beautiful results have been discovered building on the foundations laid by Vickrey and Myerson~\cite{Vickrey61, Myerson81}. 
In the past few years, there has been increasing interest in understanding how to design mechanisms for two-sided markets. In a two-sided market, the items are owned by a group of selfish agents known as sellers who form one side of the market, and the other side of the market -- the buyers wish to purchase the items from the sellers. Sellers have values for the items that they own and buyers have values for items that are on the market. Both sides are {assumed} to act strategically in order to maximize their utilities. The goal is to design a mechanism to facilitate trade between the two groups and optimize a certain objective, e.g., efficiency. Usually, the two-sided market is studied in the Bayesian model, where we assume that both the buyers' and sellers' values are drawn from some known distributions. The major difference between one-sided markets and two-sided markets is that all items can be viewed as owned by the mechanism in one-sided markets, thus the mechanism does not need to handle the strategic behavior from the seller side.

The interest on two-sided markets can largely be attributed to many important applications, such as stock exchange, online ad exchange platforms (e.g., Google's Doubleclick, Micorosoft's AdECN. etc.) where advertisers try to purchase ad slots from websites who wish to sell the slots, the FCC Spectrum Auction where the telecommunication companies try to purchase spectrum from television broadcasting companies, online market places (e.g., Amazon, eBay etc.) where buyers and sellers trade on the large-scaled trading platforms, and Uber where passengers try to book trips provided by Uber drivers.

Another crucial difference between one-sided markets and two-sided markets is that the mechanisms are usually required to be {budget-balanced} in two-sided markets. Namely, the mechanism should not gain or lose any money from the market. This is usually referred as \emph{Strong Budget Balance (SBB)}. A relaxed version of this definition is called \emph{Weak Budget Balance (WBB)}, where we only need to make sure that the mechanism does not inject money into the market. An even weaker definition known as \emph{Ex-ante Weak Budget Balance (ex-ante WBB)} only requires the expected sum of payments collected from the buyers is no less than the expected sum of the sellers' gains. All variants of the Budget Balance constraint are natural for two-sided markets, however, adding such a constraint greatly increases the difficulty of designing mechanisms. The simplest two-sided market is usually referred to as the bilateral trading, in which there is a single seller who owns an item and there is a single buyer who wants to purchase the item. Both the buyer and the seller have quasi-linear utility functions, and the buyer's value for the item $b$ and the seller's value $s$ are drawn independently from two known distributions $D^{B}$ and $D^{S}$. Unfortunately, a classical impossibility result by Myerson and Satterthwaite~\cite{MyersonS83} stated that even in bilateral trading there is no \notshow{Bayesian Incentive Compatible (BIC) ex-ante WBB mechanism}{Bayesian Incentive Compatible (BIC), Individual Rational (IR), and ex-ante WBB mechanism} that achieves full efficiency (trade whenever the buyer's value is higher than the seller's value). In the same paper, Myerson and Satterthwaite provided a mechanism that maximizes the efficiency among all mechanisms that are BIC and SBB\footnote{As it turns out (Theorem~\ref{thm:ex-ante WBB to SBB}), under interim IR and BIC constraints, ex-ante WBB is equivalent to SBB. So their mechanism is also optimal among all ex-ante WBB mechanisms.}. This is usually known as the ''second-best'' mechanism. Unfortunately, even in the simple setting of bilateral trading, the ``second-best'' mechanism has extremely complex allocation and price rules that are determined by solving a system of differential equations that depend on the buyer's and seller's distributions. It is difficult to imagine how to implement such a mechanism in practice.

Motivated by the aforementioned results, we aim to design simple IR, BIC and SBB mechanisms that approximately {maximizes} efficiency in two-sided markets. There are two standard {ways} to measure efficiency. The first one is the \emph{Social Welfare} induced by the mechanism. The second one is the \emph{Gains from Trade (GFT)}, which is the expected gain in the social welfare induced by the mechanism. For example, if the buyer has value $\$ 5$ for the item and the seller has value $\$ 3$ and they decide to trade. The social welfare is $\$ 5$, while the GFT is $\$ 2$. An astute reader might have already realized that any mechanism that maximizes the social welfare necessarily also maximizes the GFT. However, the two objectives are rather different regarding approximation. In the example above, the mechanism that does not trade still achieves a social welfare of $\$ 3$ which is $60\%$ of the optimal social welfare, but the GFT of not trading is $ 0$ which is not within any constant factor of the optimal GFT. It is not hard to observe that any constant factor approximation of GFT is necessarily a constant factor approximation of the social welfare, but the other direction does not hold. In this sense, GFT is a more difficult objective to approximate and any approximation to GFT provides a stronger guarantee to the efficiency. There are a few mechanisms that manage to provide constant factor approximation to the social welfare in surprisingly broad settings~\cite{DuttingRT14,Colini-Baldeschi16, Colini-Baldeschi16c,BlumrosenD16}. The results are sparser for GFT, McAfee provided a simple mechanism for bilateral trading that achieves $\frac{1}{2}$ of the optimal GFT if the median of the buyer's value is higher than the median of the seller's value~\cite{McAfee08}, and Blumrosen and Mizrahi provided a simple mechanism for the same setting that achieves $\frac{1}{e}$ of the optimal GFT if the buyer's value distribution has Monotone Hazard Rate (MHR).

\subsection{Our Results and Techniques}

In this paper, we provide simple mechanisms to approximate the optimal GFT obtainable by any IR, BIC and ex-ante WBB mechanisms. We believe this paper has the following three major contributions.
\begin{enumerate}
    \item 
We provide a simple IR, BIC and SBB mechanism that achieves at least $\frac{1}{2}$ of the optimal GFT in bilateral trading \emph{for arbitrary seller and buyer value distributions}.
	\item We extend the $2$-approximation to the double auction settings with \emph{arbitrary downward-closed feasibility constraints} and \emph{arbitrary seller and buyer value distributions}.
	\item We demonstrate the applicability of the Cai-Devanur-Weinberg duality framework for two-sided markets. Although our results only concern single-dimensional settings, our duality based upper bound can be easily extended to multi-dimensional settings.
\end{enumerate}

In particular, for the bilateral trading setting, we consider the following two simple IR, BIC, SBB mechanisms and show that {the better} of the two achieves $\frac{1}{2}$ of the optimal GFT.

\begin{itemize}

\item \textbf{Seller-Offering Mechanism(SOM)}: The seller offers a take-it or leave-it price for the item to the buyer. The buyer receives the item if she pays the offered price to the seller. Otherwise, the seller keeps the item and no payment is transferred. The seller chooses the price that maximizes her expected utility depending on her true value $s$ and the distribution of the buyer's value $D^{{B}}$.
\item \textbf{Buyer-Offering Mechanism(BOM)}: The buyer offers a take-it or leave-it price for the item to the seller. If the seller accepts the offered  price, the buyer receives the item and pays the seller the offered price. Otherwise, the seller keeps the item and no payment is transferred. The buyer chooses the price that maximizes her expected utility depending on her true value $b$ and the distribution of the seller's value $D^{{S}}$.
\end{itemize}

\begin{informaltheorem}\label{inforthm:bilateral trading}
	 For arbitrary buyer value distribution $D^{{B}}$ and seller value distribution $D^{{S}}$, either the SOM or the BOM achieves at least $\frac{1}{2}$ of the optimal GFT obtainable by any IR, BIC and ex-ante WBB mechanism in bilateral trading.
\end{informaltheorem}

The SOM is proposed by Blumroson and Mizrahi~\cite{BlumrosenM16}, and the BOM is the same mechanism with the roles of buyer and seller exchanged. Both mechanisms are BIC, and one might wonder whether it is possible to use a {Dominant Strategy Incentive Compatible} (DSIC) mechanism to approximate the GFT. Unfortunately, as shown by~\cite{BlumrosenD16,BlumrosenM16}, no IR, DSIC and SBB mechanism can achieve a constant fraction of the GFT\footnote{The original result was shown for the optimal GFT, but the result by Blumrosen and Mizrahi~\cite{BlumrosenM16} implies that for the hard instance the ``second-best" mechanism achieves a constant fraction of the optimal GFT. Hence, this inapproximability result extends to the best GFT obtainable by any IR, BIC and ex-ante WBB mechanisms.}, so in order to obtain a constant factor approximation for GFT one has to relax one of the three conditions. {We will see later that one can also use an IR, DSIC and  ex-ante WBB mechanism to achieve a $2$-approximation to the optimal GFT.}

 We extend our result to a more general setting that is known as the \emph{double auction} in the literature. In a double auction, there is a single type of item on the market. The buyers each want a single copy of the item and the sellers each own a single copy of the item. So, there is a single number associated with each buyer or seller's value. We assume that these values are all drawn independently from possibly different distributions. Moreover, we allow any downward-closed feasibility constraint on which buyer-seller pairs can trade. We show that for any double auction with arbitrary downward-closed feasibility constraint, one of the two IR, DSIC, ex-ante WBB mechanisms achieves $\frac{1}{2}$ of the optimal GFT obtainable by any IR, BIC and ex-ante WBB mechanisms.

 \begin{itemize}
 	\item \textbf{Generalized Seller-Offering Mechanism (GSOM)}: Given value profile $(\bold{b},\bold{s})$, assign weight $\tilde{\varphi}_i(b_i)-s_j$ to the pair of buyer $i$ and seller $j$, where $\tilde{\varphi_i}(\cdot)$ is Myerson's (ironed) virtual value function for buyer $i$~\cite{Myerson81}. Find a maximum weight matching {subject to the feasibility constraint} between the sellers and buyers according to the weights defined above. We use $\mathcal{M}^{1}(\bold{b},\bold{s})$ to denote the maximum weight matching\footnote{If there are multiple max weight matchings, we break ties lexicographically.}.  For each $(i,j)\in \mathcal{M}^{1}(\bold{b},\bold{s})$, buyer $i$ trades with seller $j$. This allocation rule is monotone and the buyer (or the seller) pays (or receives) the \emph{threshold payment}. 	
\item \textbf{Generalized Buyer-Offering Mechanism (GBOM)}: Given value profile $(\bold{b},\bold{s})$, assign weight $b_i-\tilde{\tau}_j(s_j)$ to the pair of buyer $i$ and seller $j$,  where $\tilde{\tau}_j(\cdot)$ is the (ironed) virtual value function for seller $j$\footnote{ It has a similar form as the Myerson's virtual value. Let $f^S_j$ and $F^S_j$ be the pdf and cdf for seller $j$'s value distribution, then the virtual value of seller $j$ is $s_j+\frac{F^S_j(s_j)}{f^S_j(s_j)}$.
{ If the virtual value function is not monotonically increasing, then we apply a procedure similar to the one used in~\cite{Myerson81, CaiDW16} to iron it.} See Lemma~\ref{lem:flow without ironing},~\ref{lem:canonial flow} and Appendix~\ref{appx:ironing} for more details.}. Find a maximum weight matching subject to the feasibility constraint between the sellers and buyers according to the weights defined above. We use $\mathcal{M}^{2}(\Bb,\Bs)$ to denote the maximum weight matching. For each $(i,j)\in \mathcal{M}^{2}(\Bb,\Bs)$, buyer $i$ trades with seller $j$. This allocation rule is monotone and the buyer (or the seller) pays (or receives) the \emph{threshold payment}.
 \end{itemize}

 \begin{informaltheorem}\label{inforthm:double auction}
 	In any double auction with arbitrary buyer and seller value distributions and arbitrary downward-closed feasibility constraint $\mathcal{F}$, both of {GSOM and GBOM} are IR, DSIC and ex-ante WBB. Furthermore, the better of the two mechanisms above achieves $\frac{1}{2}$ of the optimal GFT obtainable by any IR, BIC, ex-ante WBB mechanism.
 \end{informaltheorem}

It turns out there exists a generic transformation (Theorem~\ref{thm:ex-ante WBB to SBB}) that allows us to convert any IR, DSIC, ex-ante WBB mechanism to an IR, BIC, SBB mechanism without hurting the GFT. By applying this transformation to the GSOM and BSOM, we obtain the following theorem.

  \begin{informaltheorem}\label{inforthm:double auction SBB}
  In any double auction with arbitrary buyer and seller value distributions and arbitrary downward-closed feasibility constraint $\mathcal{F}$, we can design an IR, BIC, SBB mechanism that achieves $\frac{1}{2}$ of the optimal GFT obtainable by any IR, BIC, ex-ante WBB mechanism.
   \end{informaltheorem}

Our result is facilitated by applying the Cai-Devanur-Weinberg duality framework to double auctions. We first characterize the set of dual variables that has finite dual objective values as flows. For readers that are familiar with the duality framework, this sounds rather similar to the one-sided market case. Indeed the conclusion is quite similar, but the argument is different and requires us to observe a nice symmetry between the sellers and buyers. From the duality, we obtained an upper bound for the optimal GFT obtainable by any IR, BIC and ex-ante WBB mechanism. Interestingly, the upper bound has a natural format and suggests the allocation rule of our simple mechanisms. 

\subsection{More Related Work}
Our results are motivated by two different lines of work. The first one is designing simple and approximately revenue-optimal mechanisms in one-sided markets, where we have simple, deterministic and DSIC mechanisms that achieve a constant fraction of the optimal revenue obtainable by any randomized and BIC mechanisms in  fairly general multi-dimensional settings~\cite{HartlineR09, ChawlaHK07, ChawlaHMS10, HartN12, CaiH13, LiY13, BabaioffILW14, Yao15, CaiDW16, CaiZ17}. The second one is designing mechanisms that approximates the GFT. The result presented in this paper is sparser and to the best of our knowledge.~\cite{McAfee08} and~\cite{BlumrosenM16} are the only two other papers that studied this problem. In both of these papers, a constant factor approximation is obtained with respect to the optimal GFT in bilateral trading (although not achievable by any IR, BIC, and WBB mechanism) by making assumptions on the buyer and seller's value distributions. Our result obtains a constant fraction of the GFT of the ``second-best'' mechanism, but does not rely on any assumption of the value distribution and generalizes nicely to the double auction setting. It is shown by Blumrosen and Mizrahi that there is at least a constant gap between the GFT of the  ``second-best'' mechanism and the optimal GFT in bilateral trading~\cite{BlumrosenM16} when the distribution is MHR, the gap for general distributions remains open.

For the objective of social welfare, Duetting et al. \cite{DuttingRT14} considered a setting similar to ours, where the double auction can have matroid, knapsack and matching feasibility constraints over the set of buyer and seller that can be involved in the trade. They proposed a modular approach based on the deferred-acceptance heuristics from~\cite{MilgromS13}, and they obtain an IR, DSIC and WBB mechanism that achieves a constant fraction of the optimal social welfare. Recently, Colini-Baldeschi et al.~\cite{Colini-Baldeschi16} showed how to design an IR, DSIC and SBB mechanism for the same setting when the feasibility constraint is a matroid constraint. Blumrosen and Dobzinski showed that a posted price mechanism achieves $1-\frac{1}{e}$ of optimal social welfare in dominant strategies for the bilateral trading setting, and the result can be generalized to certain multi-dimensional settings. Finally, Colini-Baldeschi et al. \cite{Colini-Baldeschi16c} considered a two-sided combinatorial auctions, where the market has multiple types of items for sale. Each seller owns a few items and she has additive valuation over her items. Every buyer has XOS valuation over the items. Rather surprisingly, they showed that a variant of a sequential posted price mechanism can achieve a constant fraction of the optimal social welfare. The mechanism is simple, IR, DSIC and SBB. {A different line of work~\cite{Wilson85,SatterthwaiteW02,McAfee92, RustichiniSW94, FudenbergMZ07} showed that when the two-sided market has symmetric sellers and symmetric buyers, then the inefficiency goes away quickly in simple auctions, such as $k$-double-auction~\cite{SatterthwaiteW02} and McAffee's trade reduction mechanism~\cite{McAfee92}, when the market size grows large with symmetric sellers and symmetric buyers.}

Our proof is based on the Cai-Devanur-Weinberg duality framework, which has been used to address a number of challenging problems in one-sided markets notably gives non-trivial upper bounds for a variety of settings \cite{CaiDW16, CaiZ17, EdenFFTW16a, EdenFFTW16b}. Our paper provides an extension of the framework to two-sided markets.\\

\section{Preliminaries}\label{sec:prelim}

\noindent\textbf{Two-sided Market Settings}
\vspace{.05in}

We use $m$ to denote the number of sellers and $n$ to denote the number of buyers. For each seller $j$ (or buyer $i$), her type $s_j$ (or type $b_i$), is drawn independently from her type distribution $D^S_j$ (or $D^B_i$). Let $D^S=\times_{j=1}^m D^S_j$ and $D^B=\times_{i=1}^n D^B_i$ be the product distribution of sellers' and buyers' type profile respectively. For notational convenience, let $D^S_{-j}$ (or $D^B_{-i}$) be the distribution of types of all sellers (or buyers) except $j$ (or $i$). We use $T^S_j$ (or $T^S$, $T^B_i$, $T^B$, $T^S_{-j}$, $T^B_{-i}$) and $f^S_j$ (or $f^S$, $f^B_i$, $f^B$, $f^S_{-j}$, $f^B_{-i}$) to denote the support and density function of $D^S_j$ (or $D^S, D^S_i, D^B, D^S_{-j}, D^B_{-i}$). Let $D=(D^S, D^B)$.

\vspace{.1in}
\noindent\textbf{Double Auction Settings}
\vspace{.05in}

In this paper, we focus on a special case of two-sided markets -- double auctions, where sellers are unit-supply with homogeneous items and buyers are unit-demand. In other words, each seller only owns one item, and each buyer only wishes to buy one item and treat all the items as the same. This is a single-dimensional setting as every buyer's or seller's type can be represented as a single number. 
Let $V=\left\{(i,j)\ |\ i\in [n], j\in [m]\right\}$ be the set of all possible trading pairs between the sellers and buyers. We use $\mathcal{F}\subseteq 2^V$ to denote the \emph{feasibility constraint}. More specifically, $\mathcal{F}$ is a set system that contains all of the feasible sets of seller-buyer pairs that can be traded simultaneously. We allow any $\mathcal{F}$ that satisfies the following two properties:
\begin{itemize}
\item Every $S\in \mathcal{F}$ is a matching, in other words, for every buyer $i$ there is at most one seller $j\in[m]$ such that $(i,j)\in S$. Same for the sellers.
\item $\mathcal{F}$ is downward closed, i.e., if $S\in \mathcal{F}$, and $S' \subseteq S$, then $S'\in \mathcal{F}$.
\end{itemize}

\vspace{.05in}
\noindent\textbf{Mechanisms for Two-sided Markets}
\vspace{.05in}

Any mechanism in two-sided markets can be specified as a tuple $(A, p^B, p^S)$, where $A$ represents the allocation rule, $p^{B}$ and $p^{S}$ are the payment rule for buyers and sellers. Given a type profile $\Bs$ and $\Bb$, $A(\Bb,\Bs)\in \mathcal{F}$ is a {(random)} matching that contains all the pairs of sellers and buyers who trade with each other under this type profile. In other words, for every pairs $(i,j)\in A(\Bb,\Bs)$, buyer $i$ receives the item from seller $j$ and pays $p^B_i(\Bb,\Bs)$ to the mechanism. $p^S_j(\Bb,\Bs)$ is the amount of money seller $j$ gains from the mechanism. For notational convenience, we slightly abuse notation to let $p^B_i(b_i)=\mathbb{E}_{b_{-i},\Bs}[p^B_i(b_i,b_{-i},\Bs)]$ be buyer $i$'s  expected payment when she reports type $b_i$, over the randomness of mechanism and other agents' types. Similarly, let $p^S_j(s_j)$ be seller $j$'s expected gains when she reports type $s_j$.

In our analysis, we usually use an alternative representation of the allocation rule. For any type profile $\Bb$ and $\Bs$, for every $i\in [n]$, we use $x^B_i(\Bb,\Bs)$ to denote the probability that buyer $i$ gets an item under this type profile, i.e., the probability that $(i,j)\in A(\Bb,\Bs)$ for some $j$. Similarly, for every $j\in [m]$, we use $x^S_j(\Bb,\Bs)$ to denote the probability that seller $j$ sells her item. Given a mechanism $(x^B, x^S, p^B, p^S)$, for every type profile $(\Bb,\Bs)$, buyer $i$'s utility is $b_i\cdot x^B_i(\Bb,\Bs)-p^B_i(\Bb,\Bs)$ and similarly seller $j$'s utility is  $p^S_j(\Bb,\Bs)-s_j\cdot x^S_j(\Bb,\Bs)$.

We again slightly abuse notation to let $x^B_i(b_i)=\mathbb{E}_{b_{-i},\Bs}[x^B_i(b_i,b_{-i},\Bs)]$ be the interim probability that buyer $i$ gets an item when she reports type $b_i$, over the randomness of mechanism and other agents' types. Similarly, let $x^S_j(s_j)$ be seller $j$'s interim probability to sell her item when she reports type $s_j$.

An allocation rule $x=(x^B, x^S)$ is \emph{monotone} if for every buyer $i\in [n]$, $x^B_i(b_i,b_{-i},\Bs)$ is non-decreasing in $b_i$ for any fixed $b_{-i}$ and $\Bs$, and for every seller $j\in [m]$, $x^S_j(\Bb,s_j,s_{-j})$ is non-increasing in $s_j$ for any fixed $\Bb$ and $s_{-j}$.

An allocation rule $x=(x^B, x^S)$ is \textit{implementable} if for every $\Bs$ and $\Bb$, the allocation can be expressed as a distribution over matchings in $\mathcal{F}$. The set of all implementable allocation rule is closed and convex, which depends on the feasibility constraint $\mathcal{F}$. We denote such set as $P(\mathcal{F})$.

\vspace{.15in}
\noindent\textbf{Budget Balance Constraints}
\vspace{.05in}

There are a few variants of the budget balance constraints.
\begin{itemize}
\item \textbf{Strong Budget Balance (SBB)}: Under any type profile, the sum of all buyers' (expected) payment is equal to the sum of all sellers' (expected) gains, over the randomness of mechanism.
\item \textbf{Weak Budget Balance (WBB)}: Under any type profile, the sum of all buyers' (expected) payment is at least the sum of all sellers' (expected) gains, over the randomness of mechanism.
\item \textbf{Ex-ante Strong Budget Balance (Ex-ante SBB)}: The sum of all buyers' expected payment is equal to the sum of all sellers' expected gains, over the randomness of mechanism and the type profile of all agents.
\item \textbf{Ex-ante Weak Budget Balance (Ex-ante WBB)}: The sum of all buyers' expected payment is at least the sum of all sellers' expected gains, over the randomness of mechanism and the type profile of all agents.
\end{itemize}
We discuss the connections between these variants in Section~\ref{sec:transformation}.

\vspace{.15in}
\noindent\textbf{Gains from Trade}
\vspace{.05in}

\notshow{Gains from Trade(GFT) is the objective function we study. It describes the change of efficiency after running the mechanism. In other words, the gains from trade is the social welfare minus the sum of sellers' value. Formally, given a mechanism $M$, the expected Gains from Trade for the mechanism is written as}
{Gains from Trade(GFT) is the objective function we study. It describes the gains of social welfare induced by the mechanism. In other words, GFT is the social welfare of the allocation selected by the mechanism minus the sum of sellers' value. Formally, given a mechanism $M=(A, p^B, p^S)$, the expected GFT for the mechanism is defined as}
\begin{equation}
\gft(M)=\sum_{\Bb\in T^B}\sum_{\Bs\in T^S}f^B(\Bb)f^S(\Bs)\cdot \sum_{(i,j)\in A(\Bb,\Bs)}(b_i-s_j)
\end{equation}
or using the definition of $x^B, x^S$,
\begin{equation}
\begin{aligned}
\gft(M)=\sum_{\Bb\in T^B}\sum_{\Bs\in T^S}f^B(\Bb)f^S(\Bs)\cdot \left(\sum_{i=1}^nx^B_i(\Bb,\Bs)\cdot b_i-\sum_{j=1}^mx^S_j(\Bb,\Bs)\cdot s_j\right)
\end{aligned}
\end{equation}
We use $\opt$ to denote the highest expected GFT obtainable by any IR, BIC, ex-ante WBB mechanism. 

\vspace{.15in}
\noindent\textbf{Continuous vs. Discrete Distributions}
\vspace{.05in}

We explicitly assume that the input distributions are discrete. Nevertheless, if the distributions are continuous, we can write a continuous LP and derive the duality based upper bound similarly. All the rest of the analysis holds for continuous distributions without much modification. So our results also apply to  continuous distributions.

\section{Duality}
{
In this section, we apply the Cai-Devanur-Weinberg duality framework to obtain a benchmark for the optimal GFT obtainable by any IR, BIC and ex-ante WBB mechanism. The same formulation can easily be extended to the general two-sided markets. The main theorem for this section is shown as follows:

\begin{theorem}\label{cor:duality upper bound}
Given a double auction, let $\opt$ be the optimal GFT among all IR, BIC and ex-ante WBB mechanisms, for any {$\alpha\geq 0$},
\[ \opt\leq \max_{x\in P(\mathcal{F})}~~\sum_{\Bb,\Bs}f^B(\Bb)f^S(\Bs)\left(\sum_{i=1}^n x^B_i(\Bb,\Bs)\cdot (b_i+\alpha\cdot \tilde{\varphi}_i(b_i))-\sum_{j=1}^m x^S_j(\Bb,\Bs)\cdot (s_j+\alpha\cdot \tilde{\tau}_j(s_j))\right).\]

 $\tilde{\varphi}_i(\cdot)$ is Myerson's ironed virtual value function for buyer $i$. $\tilde{\tau}_j(\cdot)$ is symmetric to $\tilde{\varphi}_i(\cdot)$, and we will refer to it as the Myerson's ironed virtual value function for seller $j$. See Appendix~\ref{appx:ironing} for more details.
\end{theorem}

}


First, let us state an interesting equivalence among the ex-ante WBB constraint and the SBB constraint for IR and BIC mechanisms. {A SBB mechanism is clearly ex-ante WBB, and we demonstrate the other direction in Theorem~\ref{thm:ex-ante WBB to SBB}.}
{The intuition behind Theorem~\ref{thm:ex-ante WBB to SBB} is that if our mechanism has positive surplus, we can simply divide the surplus to the agents evenly and independently from their reported types (Lemma~\ref{lem:ewbb to ebb}). Now we have an ex-ante SBB mechanism. Next, we massage the payments so that all interim payments remain unchanged, while under every type profile the sum of buyers' payments equal the sum of sellers' gains. This is achieved via an interesting linear transformation on the payments (Theorem~\ref{thm:ebb to sbb}).}

\begin{theorem}\label{thm:ex-ante WBB to SBB}
	Given an IR, BIC, ex-ante WBB mechanism $M=(x^B, x^S, p^B, p^S)$ with non-negative payment rule, there exists another non-negative payment rule $(p^{B'}, p^{S'})$ such that the mechanism $M'=(x^B, x^S, p^{B'}, p^{S'})$ is IR, BIC and SBB.
\end{theorem}
\begin{proof}
It follows from Lemma~\ref{lem:ewbb to ebb} and Theorem~\ref{thm:ebb to sbb}. See Section~\ref{sec:transformation} for more details.
\end{proof}

\subsection{Duality Framework}
We aim to find an IR, BIC, ex-ante WBB mechanism that maximizes the GFT in a double auction. Due to Theorem~\ref{thm:ex-ante WBB to SBB}, this is equivalent to finding an IR, BIC, SBB mechanism that maximizes the GFT, which is captured by the following LP (Figure~\ref{fig:primal}). To ease notation, we use a special type $\varnothing$ to represent the choice of not participating in the mechanism. More specifically, $x^B_i(\varnothing,b_{-i},\Bs)=p^B_i(\varnothing,b_{-i},\Bs)=0$ for any $b_{-i}$ and $\Bs$. Now the interim IR constraint can be described as another BIC constraint: for any type $b_i$, buyer $i$ does not want to report type $\varnothing$. Let $T_i^{B+}=T^B_i\cup \{\varnothing\}$. Similarly, we define the same type for every seller $j$ and let $T_j^{S+}=T^S_j\cup \{\varnothing\}$.

\begin{figure}[ht]\label{fig:primal}
\colorbox{MyGray}{
\begin{minipage}{\textwidth} {
\noindent\textbf{Variables:}
\begin{itemize}
\item $p^B_i(\Bb,\Bs), p^S_j(\Bb,\Bs)$, for all buyers $i\in [n]$, sellers $j\in [m]$, and value profile $\Bb\in T^B, \Bs\in T^S$ denoting buyer $i$'s payment and seller $j$'s gains accordingly, under type profile $(\Bb,\Bs)$.
\item $x^B_i(\Bb,\Bs), x^S_j(\Bb,\Bs)$, for all buyers $i\in [n]$, sellers $j\in [m]$, and value profile $\Bb\in T^B, \Bs\in T^S$ denoting the probability that buyer $i$ gets an item and seller $j$ sells her item accordingly, under type profile $(\Bb,\Bs)$.
\end{itemize}
\textbf{Constraints:}
\begin{itemize}
\item $\mathbb{E}_{b_{-i},\Bs}\left[\left(b_i\cdot x^B_i(b_i,b_{-i},\Bs)-p^B_i(b_i,b_{-i},\Bs)\right)-\left(b_i'\cdot x^B_i(b_i',b_{-i},\Bs)-p^B_i(b_i',b_{-i},\Bs)\right)\right]\geq 0$, for all buyer $i$, and types $b_i\in T^B_i, b_i'\in T^{B+}_i$, guaranteeing that the mechanism is BIC and IR for all buyers.
\item $\mathbb{E}_{\Bb,s_{-j}}\left[\left(p^S_j(\Bb,s_j,s_{-j})-s_j\cdot x^S_j(\Bb,s_j,s_{-j})\right)-\left(p^S_j(\Bb,s_j',s_{-j})-s_j'\cdot x^S_j(\Bb,s_j',s_{-j})\right)\right]\geq 0$, for all seller $j$, and types $s_j\in T^S_j, s_j'\in T^{S+}_j$, guaranteeing that the mechanism is BIC and IR for all sellers.
\item $\sum_{i=1}^n p^B_i(\Bb,\Bs)=\sum_{j=1}^m p^S_j(\Bb,\Bs)$ for every $\Bb\in T^B, \Bs\in T^S$, guaranteeing that the mechanism is SBB.
\item $x \in P(\mathcal{F})$, guaranteeing $x=(x^B, x^S)$ is implementable.
\end{itemize}
\textbf{Objective:}
\begin{itemize}
\item $\max \mathbb{E}_{\Bb,\Bs}\left[\sum_{i=1}^nx^B_i(\Bb,\Bs)\cdot b_i-\sum_{j=1}^mx^S_j(\Bb,\Bs)\cdot s_j\right]$, the expected Gains from Trade.\\
\end{itemize}}
\end{minipage}}
\caption{A Linear Program (LP) for Maximizing Gains from Trade.}
\label{fig:primal}
\end{figure}

We take the partial Lagrangian dual of the LP in Figure~\ref{fig:primal} by lagrangifying the BIC constraints for all agents. Let $\lambda^B_i(b_i,b_i')$ be the Lagrange multiplier associated with buyer $i$'s BIC constraints, and $\lambda^S_j(s_j,s_j')$ be the Lagrange multiplier associated with seller $j$'s BIC constraints (see Figure~\ref{fig:Lagrangian}).

\begin{figure}[ht]
\colorbox{MyGray}{
\begin{minipage}{\textwidth} {
\noindent\textbf{Variables:}
\begin{itemize}
\item $\lambda^B_i(b_i,b_i')$ for all $i, b_i\in T^B_i, b_i'\in T^{B+}_i$, the Lagrangian multiplier for buyers' BIC and IR constraints.
\item $\lambda^S_j(s_j,s_j')$ for all $j, s_j\in T^S_j, s_j'\in T^{S+}_j$, the Lagrangian multiplier for sellers' BIC and IR constraints.
\item $p^B_i(\Bb,\Bs), p^S_j(\Bb,\Bs), x^B_i(\Bb,\Bs), x^S_j(\Bb,\Bs)$ for all $i, j, \Bb\in T^B, \Bs\in T^S$.
\end{itemize}
\textbf{Constraints:}
\begin{itemize}
\item $\lambda^B_i(b_i,b_i')\geq 0$ for all $i, b_i\in T^B_i, b_i'\in T^{B+}_i$.
\item $\lambda^S_j(s_j,s_j')\geq 0$ for all $j, s_j\in T^S_j, s_j'\in T^{S+}_j$.
\item $\sum_{i=1}^n p^B_i(\Bb,\Bs)=\sum_{j=1}^m p^S_j(\Bb,\Bs)$ for every $\Bb\in T^B, \Bs\in T^S$.
\item $x \in P(\mathcal{F})$.
\end{itemize}
\textbf{Objective:}
\begin{itemize}
\item $\min_{\lambda}\max_{x, p} \L(\lambda, x, p)$.\\
\end{itemize}}
\end{minipage}}
\caption{Partial Lagrangian of the LP for Maximizing Gains from Trade.}
\label{fig:Lagrangian}
\end{figure}

Here $\L(\lambda, x, p)$ is the partial Lagrangian that is defined as:
\begin{equation}\label{equ:language function}
\begin{aligned}
&\L(\lambda, x, p)\\
&=\mathbb{E}_{\Bb,\Bs}\left[\sum_{i=1}^nx^B_i(\Bb,\Bs)\cdot b_i-\sum_{j=1}^mx^S_j(\Bb,\Bs)\cdot s_j\right]\\
&+\sum_i\sum_{b_i,b_i'}\lambda^B_i(b_i,b_i')\cdot \mathbb{E}_{b_{-i},\Bs}\left[\left(b_i\cdot x^B_i(b_i,b_{-i},\Bs)-p^B_i(b_i,b_{-i},\Bs)\right)-\left(b_i'\cdot x^B_i(b_i',b_{-i},\Bs)-p^B_i(b_i',b_{-i},\Bs)\right)\right]\\
&+\sum_j\sum_{s_j,s_j'}\lambda^S_j(s_j,s_j')\cdot \mathbb{E}_{\Bb,s_{-j}}\left[\left(p^S_j(\Bb,s_j,s_{-j})-s_j\cdot x^S_j(\Bb,s_j,s_{-j})\right)-\left(p^S_j(\Bb,s_j',s_{-j})-s_j'\cdot x^S_j(\Bb,s_j',s_{-j})\right)\right]\\
\end{aligned}
\end{equation}

Similar to the framework in~\cite{CaiDW16}, we show that the dual variables need to have certain nice format in order for the partial Lagrangian to be finite. Notice that the argument here is different from the one in one-sided markets, as the payments are no longer unconstrained and have to satisfy the SBB constraint. We call a payment rule $p=(p^B, p^S)$ \textit{balanced} if it satisfies the SBB constraint.

\begin{equation}\label{equ:language function}
\begin{aligned}
\L(\lambda, x, p)&=\sum_i\sum_{b_i\in T^B_i}\bigg( \sum_{b_i'\in T^{B+}_i}\lambda^B_i(b_i,b_i')-\sum_{b_i'\in T^{B}_i}\lambda^B_i(b_i',b_i) \bigg)\cdot \mathbb{E}_{b_{-i},\Bs}\left[p^B_i(b_i,b_{-i},\Bs)\right]\\
&+\sum_j\sum_{s_j\in T^S_j}\bigg(\sum_{s_j'\in T^{S}_j}\lambda^S_j(s_j',s_j)-\sum_{s_j'\in T^{S+}_j}\lambda^S_j(s_j,s_j') \bigg)\cdot \mathbb{E}_{\Bb,s_{-j}}\left[p^S_j(\Bb,s_j,s_{-j})\right]\\
&+\sum_{\Bb}\sum_{\Bs}f^B(\Bb)f^S(\Bs)\cdot \left(\sum_{i=1}^nx^B_i(\Bb,\Bs)\cdot b_i-\sum_{j=1}^mx^S_j(\Bb,\Bs)\cdot s_j\right)\\
&+\sum_i\sum_{b_i,b_i'}\lambda^B_i(b_i,b_i')\cdot \mathbb{E}_{b_{-i},\Bs}\left[\left(b_i\cdot x^B_i(b_i,b_{-i},\Bs)\right)-\left(b_i'\cdot x^B_i(b_i',b_{-i},\Bs)\right)\right]\\
&+\sum_j\sum_{s_j,s_j'}\lambda^S_j(s_j,s_j')\cdot \mathbb{E}_{\Bb,s_{-j}}\left[\left(s_j'\cdot x^S_j(\Bb,s_j',s_{-j})\right)-\left(s_j\cdot x^S_j(\Bb,s_j,s_{-j})\right)\right]\\
\end{aligned}
\end{equation}

\begin{definition}[Useful Dual Variables~\cite{CaiDW16}]
A set of feasible duals $\lambda$ is \textbf{useful} if
\[\max_{x \in P(\mathcal{F}),~p~\text{balanced}} \L(\lambda, x, p)< \infty.\]
\end{definition}

\begin{lemma}\label{lem:useful dual}
A set of dual variables is useful iff there exists a nonnegative number 
{$\alpha$} such that for every $i, j$ and every $b_i\in T^B_i, s_j\in T^S_j$,
\begin{equation}\label{equ:def of alpha}
\frac{1}{f^B_i(b_i)}(\sum_{b_i'\in T^{B+}_i}\lambda^B_i(b_i,b_i')-\sum_{b_i'\in T^{B}_i}\lambda^B_i(b_i',b_i))=\frac{1}{f^S_j(s_j)}(\sum_{s_j'\in T^{S+}_j}\lambda^S_j(s_j,s_j')-\sum_{s_j'\in T^{S}_j}\lambda^S_j(s_j',s_j))=\alpha
\end{equation}

Equivalently, it means for each buyer $i$ (or seller $j$), $\lambda^B_i$ (or $\lambda^S_j$) forms a valid flow on the following graph. We use buyer $\lambda^B_i$ as an example to illustrate this flow:
\begin{itemize}
\item Nodes: A super source $S$, a super sink $\varnothing$, and a node $b_i$ for every type $b_i\in T^B_{i}$.
\item An edge from $S$ to $b_i$ of weight $\alpha\cdot f^B_i(b_i)$, for all $b_i\in T^B_{i}$.
\item An edge from $b_i$ to $b_i'$ of weight $\lambda^B_i(b_i,b_i')$ for all $b_i\in T^B_{i}$, and $b_i'\in T^{B+}_{i}$ (including the sink).
\end{itemize}
Equation~\eqref{equ:def of alpha} guarantees that we have flow conservation on every node of this graph.
\end{lemma}
\begin{proof}
For every $i\in [n], j\in [m]$ denote
\[A_i(b_i)=\frac{1}{f^B_i(b_i)}\left(\sum_{b_i'\in T^{B+}_i}\lambda^B_i(b_i,b_i')-\sum_{b_i'\in T^{B}_i}\lambda^B_i(b_i',b_i)\right)\]
\[B_j(s_j)=\frac{1}{f^S_j(s_j)}\left(\sum_{s_j'\in T^{S+}_j}\lambda^S_j(s_j,s_j')-\sum_{s_j'\in T^{S}_j}\lambda^S_j(s_j',s_j)\right)\]
If there exists $i,j,b_i,s_j$ such that $A_i(b_i)\not=B_j(s_j)$. WLOG, assume $A_i(b_i)>B_j(s_j)$. We choose some arbitrary $b_{-i}, s_{-j}$ and increase $p_i^B(\Bb,\Bs)$ and $p_j^S(\Bb,\Bs)$ simultaneously by $\zeta$, while keeping the payment for other type profiles unchanged.  The payment rule stays balanced and with Equation~\eqref{equ:language function}, $\L(\lambda, x, p)$ will change by
\begin{equation*}
\begin{aligned}
\Delta\L(\lambda, x, p)&=\zeta\cdot f^B_i(b_i)A_i(b_i)\cdot f^B_{-i}(b_{-i})f^S(\Bs)-\zeta\cdot f^S_j(s_j)B_j(s_j)\cdot f^B(\Bb)f^S_{-j}(s_{-j})\\
&=\zeta\cdot f^B(\Bb)f^S(\Bs)(A_i(b_i)-B_j(s_j))>0
\end{aligned}
\end{equation*}
If we let $\zeta\to \infty$, clearly the payment rule is still balanced, but $\L(\lambda,x,p)$ goes to infinity. Hence, $\lambda$ is not useful. Thus for every $i,j$ and $b_i,s_j$, $A_i(b_i)=B_j(s_j)$. There exists an $\alpha$ such that Equation~(\ref{equ:def of alpha}) holds. Moreover, we can view Equation~\eqref{equ:def of alpha} as the flow conservation condition for $b_i$ (or $s_j$), therefore $\lambda_i^B$ (or $\lambda_j^S$) must form a flow.

Next, we argue that $\alpha$ must be nonnegative. Note that $$\sum_{b_i\in T_i^B} \lambda^B_i(b_i,\emptyset) = \sum_{b_i\in T_i^B}\left(\sum_{b_i'\in T_i^{B+}} \lambda^B(b_i,b_i')-\sum_{b_i'\in T_i^{B}} \lambda^B(b'_i,b_i) \right) = \alpha\cdot \sum_{b_i\in T_i^B}f_i^B(b_i).$$ Since $\lambda_i^B(\emptyset, b_i)$ is nonnegative for all type $b_i$, $\alpha$ must also be nonnegative.

{For the other direction, if $A_i(b_i)=B_j(s_j)$ for any $i,j,b_i,s_j$, $L(\lambda,x,p)$ only depends on $x$. Since $x$ is bounded, the maximum of the partial Lagrangian is finite.}
\end{proof}

\vspace{.05in}
According to Lemma~\ref{lem:useful dual}, for useful dual variables, we replace $\sum_{b'_i\in T^{B+}_i}\lambda^B_i(b_i,b_i')$ with $\sum_{b'_i\in T^{B}_i}\lambda^B_i(b_i',b_i)+\alpha\cdot f^B_i(b_i)$ and $\sum_{s'_j\in T^{S+}_j}\lambda^S_j(s_j,s_j')$ with $\sum_{s'_j\in T^{S}_j}\lambda^S_j(s_j',s_j)+\alpha\cdot f^S_j(s_j)$ for every $i, j$ in the partial Lagrangian $\L(\lambda, x, p)$. After simplification, we have
\begin{equation}\label{equ:lagrangian with phi}
\L(\lambda, x, p)=\sum_{\Bb,\Bs}f^B(\Bb)f^S(\Bs)\left(\sum_{i=1}^n x^B_i(\Bb,\Bs)\cdot \Phi^B_i(b_i)-\sum_{j=1}^m x^S_j(\Bb,\Bs)\cdot \Phi^S_j(s_j)\right)
\end{equation}
where $\Phi^B_i(b_i)$ and $\Phi^S_j(s_j)$ are defined in Definition~\ref{def:virtual value}.

\begin{definition}[Virtual Value Function]\label{def:virtual value}
For any $\lambda$ that satisfies 
Equation~\eqref{equ:def of alpha} for all buyer $i$ and her type $b_i$ and seller $j$ and her type $s_j$, we define a corresponding virtual value function $\Phi^B_i(\cdot)$, such that for every type $b_i\in T^B_i$,
\[\Phi^B_i(b_i)=(\alpha+1)\cdot b_i-\frac{1}{f^B_i(b_i)}\sum_{b_i'\in T^B_i}\lambda^B_i(b_i',b_i)\cdot(b_i'-b_i)\]

\noindent Similarly for every seller $j$, define $\Phi^S_j(\cdot)$ such that for every type $s_j\in T^S_j$,
\[\Phi^S_j(s_j)=(\alpha+1)\cdot s_j-\frac{1}{f^S_j(s_j)}\sum_{s_j'\in T^S_j}\lambda^S_j(s_j',s_j)\cdot(s_j'-s_j)\]
\end{definition}

\subsection{The Canonical Flow}~\label{sec:canonical flow}
In this section we describe a flow and using which we derive an upper bound for the GFT. For buyers and sellers, we use different flows. For every $i$, $\lambda^B_i(b_i',b_i)>0$ if and only if $b_i'$ is the predecessor type of $b_i$, i.e., the smallest value in the support set $T^B_i$ that is greater than $b_i$. For every $j$, $\lambda^S_j(s_j',s_j)>0$ if and only if $s_j'$ is the successor type of $s_j$, i.e., the largest value in the support set $T^S_j$ that is smaller than $s_j$. In other words, for buyer, the flow goes from higher types to lower types while for the seller, the flow goes from lower types to higher ones. Our flow $\lambda$ has the following property.

\begin{lemma}\label{lem:flow without ironing}
{For any $\alpha\geq 0$}, there exists a set of dual variables such that:
\begin{itemize}
	\item for every buyer $i$ and every type $b_i$, \[\Phi^B_i(b_i)= b_i+\alpha\cdot \left(b_i-\frac{\sum_{t>b_i}f^B_i(t)\cdot (\hat{b}_i-b_i)}{f^B_i(b_i)}\right),\] where $\hat{b}_i$ is the predecessor type of $b_i$;
	\item for every seller $j$ and every type $s_j$, \[\Phi^S_j(s_j)= s_j+\alpha\cdot \left(s_j-\frac{\sum_{t<s_j}f^S_j(t)\cdot {(\hat{s}_j-{s}_j)}}{f^S_j(s_j)}\right),\] where $\hat{s}_j$ is the successor type of $s_j$.
\end{itemize}
 Let $\varphi_i(b_i):= b_i-\frac{\sum_{t>b_i}f^B_i(t)\cdot (\hat{b}_i-b_i)}{f^B_i(b_i)}$ and $\tau_j(s_j):= s_j-\frac{\sum_{t<s_j}f^S_j(t)\cdot {(\hat{s}_j-{s}_j)}}{f^S_j(s_j)}$.
 $\varphi_i(\cdot)$ is Myerson's virtual value function for buyer $i$. As $\tau_j(\cdot)$ is symmetric to $\varphi_i(\cdot)$, we will refer to it as the Myerson's virtual value function for seller $j$. Furthermore, if $\varphi_i(\cdot)$ (or $\tau_j(\cdot)$) is monotonically non-decreasing, we say $D^B_i$ (or $D^S_j$) is regular.
\end{lemma}

\begin{proof}
For each buyer $i$ and her type $b_i$, $b_i$ only gets flow from $\hat{b}_i$, which is originally the flow from source to a type $b_i'>b_i$. The total amount of flow goes into $b_i$ is $\alpha\cdot \sum_{b_i'>b_i}f_i^B(b_i')$. By Definition~\ref{def:virtual value},
\[\Phi^B_i(b_i)=(\alpha+1)\cdot b_i-\frac{\alpha\cdot \sum_{b_i'>b_i}f_i^B(b_i')\cdot (\hat{b}_i-b_i)}{f_i^B(b_i)}=b_i+\alpha\cdot \varphi_i(b_i)\]
Similarly, for each seller $j$ and her type $s_j$, $s_j$ only gets flow from $\hat{s}_j$, and total amount of flow is $\alpha\cdot \sum_{s_j'<s_j}f_j^S(s_j')$.
\[\Phi^S_j(s_j)=(\alpha+1)\cdot s_j-\frac{\alpha\cdot \sum_{s_j'<s_j}f_j^S(s_j')\cdot (\hat{s}_j-s_j)}{f_j^S(s_j)}=s_j+\alpha\cdot \tau_j(s_j)\]
\end{proof}

When the distributions are irregular, we can iron the virtual value function with the method shown in~\cite{CaiDW16} by adding loops to the flow. {See Appendix~\ref{appx:ironing} for more details about the ironing process and the proof of Lemma~\ref{lem:canonial flow}.}

\begin{lemma}\label{lem:canonial flow}
For any $\alpha\geq 0$, there exists a set of dual variables $\lambda$ such that
\[\Phi^B_i(b_i)= b_i+\alpha\cdot \tilde{\varphi}_i(b_i)\]
\[\Phi^S_j(s_j)= s_j+\alpha\cdot \tilde{\tau}_j(s_j)\]
\end{lemma}

\hspace{0.1in}

Now we are ready to prove our main theorem of this section.

\begin{prevproof}{Theorem}{cor:duality upper bound}
It follows directly from Equation~\eqref{equ:lagrangian with phi} and Lemma~\ref{lem:canonial flow}.
\end{prevproof}


\section{Single Buyer, Single Seller Bilateral Trading}\label{sec:single}
To warm up, we first study the classical bilateral trading setting, when there is only one buyer and one seller. In this section, we show two simple IR, BIC, SBB mechanisms and prove that the better one obtains at least half of the optimal GFT.

\begin{enumerate}
\item \textbf{Seller-Offering Mechanism(SOM)}: The seller posts a take-it or leave-it price $q_B(s)$ for the item to the buyer. The item price depends on her true type $s$ and the buyer's value distribution. The buyer has to pay $q_B(s)$ to the seller if she chooses to take the item.
\item \textbf{Buyer-Offering Mechanism(BOM)}: The buyer posts a take-it or leave-it price $q_S(b)$ for the item to the seller. The item price depends on her true type $b$ and the seller's value distribution. The seller can get $q_S(b)$ from the buyer if she chooses to sell the item.
\end{enumerate}

Both mechanisms are clearly IR and SBB. For SOM, the seller has the information of her true type $s$ and the buyer's type distribution $D^B$. She will choose $q_B(s)$ to maximize her expected utility and thus the mechanism is BIC for the seller. Then the buyer sees the posted price and will buy if and only if her value $b$ is greater than $q_B(s)$, no matter what $s$ is. Hence the mechanism is DSIC for the buyer. Similarly, BOM is BIC for the buyer and DSIC for the seller. Let $\som$ and $\bom$ be the expected GFT for SOM and BOM. In the following Lemma, we characterize the value of $\som$ and $\bom$.

\begin{lemma}\label{lem:single_ppm}
\[\som=\mathbb{E}_{b,s}[(b-s)\cdot \ind[\tilde{\varphi}(b)> s]]\]
\[\bom=\mathbb{E}_{b,s}[(b-s)\cdot \ind[b>\tilde{\tau}(s)]]\]
\end{lemma}

\begin{proof}
 If the seller uses $q$ as the posted price in the SOM, the seller's expected utility is $(q-s)\cdot \Pr_{b\sim D^B}[b\geq q]$\footnote{We assume here the buyer will buy the item with zero utility. It makes no difference from buying the item with strict positive utility since the seller can decrease the posted price for an arbitrarily small constant and the seller's utility will be changed as small as possible.}. And $q_S(s)=\argmax_q \left((q-s)\cdot \Pr_{b\sim D^B}[b\geq q]\right)$. Let $T^B=\{\beta_1,...,\beta_{K}\}$ where $\beta_1<\beta_2<...<\beta_{K}$. Clearly the optimal item price $q_B(s)$ must take one of the values from support set $T^B$, i.e., $q_S(s)=\beta_k$ where $k$ maximizes $(\beta_k-s)\cdot \sum_{r=k}^{K} f^B(\beta_r)$.
{The following Lemma characterizes the seller's gains, in function of the buyer's ironed virtual value.} The proof is postponed to the Appendix~\ref{sec:proof single}.

\begin{lemma}\label{lem:myerson-phi}
Let $T^B=\{\beta_1,...,\beta_{K}\}$ where $\beta_1<\beta_2<...<\beta_{K}$. For any $k\in [K]$,
\[\sum_{r=k}^{K}\tilde{\varphi}(\beta_r)\cdot f^B(\beta_r)\geq\beta_k\cdot \sum_{r=k}^{K} f^B(\beta_r).\]
The equality holds if type $\beta_k$ is not in the interior of any ironed interval.
\end{lemma}

According to Lemma~\ref{lem:myerson-phi}, $(\beta_k-s)\cdot \sum_{r=k}^{K} f^B(\beta_r)\leq\sum_{r=k}^{K} (\tilde{\varphi}(\beta_r)-s)\cdot f^B(\beta_r)$ for all $k\in [K]$.  {Moreover, equality holds for the optimal $k$. That is,} to maximize $\sum_{r=k}^{K} (\tilde{\varphi}(\beta_r)-s)\cdot f^B(\beta_r)$, since $\tilde{\varphi}(\beta_r)$ is non-decreasing in $r$, the optimal $k^{*}$ is the smallest $k$ such that $\tilde{\varphi}(\beta_k)>s$. Clearly, $\beta_{k^{*}}$ is not in the interior of any ironed interval because $\tilde{\varphi}(\beta_{k^*-1})\leq s < \tilde{\varphi}(\beta_{k^*})$.
Thus, $(\beta_{k^{*}}-s)\cdot \sum_{r=k}^{K} f^B(\beta_r)=\sum_{r=k^{*}}^{K} (\tilde{\varphi}(\beta_r)-s)\cdot f^B(\beta_r)$ and the seller should post $\beta_{k^{*}}$ as the price. In SOM, the trade will be made if and only if $b\geq \beta_{k^{*}}$, which is equivalent to $\tilde{\varphi}(b)>s$. Hence, $\som=\mathbb{E}_{b,s}[(b-s)\cdot \ind[\tilde{\varphi}(b)>s]]$.

Similarly, if the buyer chooses $q$ as the posted price in BOM, the buyer's expected utility is $(b-q)\cdot \Pr_{s\sim D^S}[s\leq q]$. Let $T^S=\{\alpha_1,...,\alpha_{K'}\}$ where $\alpha_1<\alpha_2<...<\alpha_{K'}$. The optimal item price $q_S(b)$ must take one value from support set $T^S$, i.e., $q_B(b)=\alpha_k$ where $k$ maximizes $(b-\alpha_k)\cdot \sum_{r=1}^{k} f^S(\alpha_r)$. {The following Lemma characterizes the buyer's cost, in function of the seller's ironed virtual value.} The proof is postponed to the Appendix~\ref{sec:proof single}.

\begin{lemma}\label{lem:myerson-tau}
Let $T^S=\{\alpha_1,...,\alpha_{K'}\}$ where $\alpha_1<\alpha_2<...<\alpha_{K'}$. For any $k\in [K']$,
\[\sum_{r=1}^{k}\tilde{\tau}(\alpha_r)\cdot f^S(\alpha_r)\leq\alpha_k\cdot \sum_{r=1}^{k} f^S(\alpha_r)\]
The equality holds if type $\alpha_k$ is not in the interior of an ironed interval.
\end{lemma}

We use an argument similar to the one for SOM.
{From Lemma~\ref{lem:myerson-tau} we know that $(b-\alpha_k)\cdot \sum_{r=1}^{k} f^S(\alpha_r)\leq\sum_{r=1}^{k} (b-\tilde{\tau}(\alpha_r))\cdot f^S(\alpha_r)$. Since $\tilde{\tau}(\alpha_r)$ is non-decreasing on $r$, the optimal $k^{\dagger}$ that maximizes $\sum_{r=1}^{k} (b-\tilde{\tau}(\alpha_r))\cdot f^S(\alpha_r)$}
is the largest $k$ such that $\tilde{\tau}(\alpha_k)<b$. For the same reason, $\alpha_{k^{\dagger}}$ cannot lie in the interior of any ironed interval, hence equality holds and the best posted price for the buyer is $\alpha_{k^{\dagger}}$. In BOM, the trade will be made if and only if $s\leq \alpha_{k^{\dagger}}$, which is equivalent to $b>\tilde{\tau}(s)$. Hence, $\bom=\mathbb{E}_{b,s}[(b-s)\cdot \ind[b>\tilde{\tau}(s)]]$.
\end{proof}

\vspace{.1in}
With Lemma~\ref{lem:single_ppm}, we are ready to show that the better of SOM and BOM has GFT at least $\frac{1}{2}\opt$.

\begin{theorem}\label{thm:single}
$$\som+\bom\geq \opt$$
\end{theorem}
\begin{proof}

By Theorem~\ref{cor:duality upper bound},
\[\opt\leq \max_{x\in P(\mathcal{F})}\sum_b\sum_s f^B(b)f^S(s)\left(x^B(b,s)\cdot(b+\tilde{\varphi}(b))-x^S(b,s)\cdot(s+\tilde{\tau}(s))\right)\]

Notice that $x^B(b,s)=x^S(b,s)$ for all $b,s$, the maximum value is achieved when $x^B(b,s)=1$ if $b+\tilde{\varphi}(b)-s-\tilde{\tau}(s)>0$ and $x^B(b,s)=0$ otherwise. In other words,
\begin{equation}\label{equ:upper bound of gft}
\opt\leq \sum_b\sum_s f^B(b)f^S(s)\cdot \left(b+\tilde{\varphi}(b)-s-\tilde{\tau}(s)\right)\cdot \ind[b+\tilde{\varphi}(b)-s-\tilde{\tau}(s)> 0]
\end{equation}

We further relax the RHS of Inequality~(\ref{equ:upper bound of gft}),
\begin{equation}\label{equ:gft-separate}
\begin{aligned}
\opt&\leq \mathbb{E}_{b,s}\left[\left((\tilde{\varphi}(b)-s)+(b-\tilde{\tau}(s))\right)\cdot \ind[b+\tilde{\varphi}(b)-s-\tilde{\tau}(s)> 0]\right]\\
&\leq \mathbb{E}_{b,s}\left[(\tilde{\varphi}(b)-s)\cdot \ind[\tilde{\varphi}(b)-s>0]\right]+\mathbb{E}_{b,s}\left[(b-\tilde{\tau}(s))\cdot \ind[b-\tilde{\tau}(s)>0]\right]
\end{aligned}
\end{equation}

Notice that for any $s$, let $k$ be the smallest value such that $\tilde{\varphi}(\beta_k)-s>0$.
\begin{equation}
\begin{aligned}
\mathbb{E}_{b}[(\tilde{\varphi}(b)-s)\cdot \ind[\tilde{\varphi}(b)-s>0]]
&=\sum_{r=k}^K\tilde{\varphi}(\beta_r)\cdot f^B(\beta_r)-s\cdot \Pr_{b}[\tilde{\varphi}(b)-s>0]\\
&=\beta_k\cdot \sum_{r=k}^K f^B(\beta_r)-s\cdot \Pr_{b}[\tilde{\varphi}(b)-s>0]\\
&\leq \sum_{r=k}^K \beta_r\cdot f^B(\beta_r)-s\cdot \Pr_{b}[\tilde{\varphi}(b)-s>0]\\
&=\mathbb{E}_{b}[(b-s)\cdot \ind[\tilde{\varphi}(b)-s>0]]
\end{aligned}
\end{equation}
where the second equation is because $\beta_k$ must not lie in the interior of any ironed interval, so Lemma~\ref{lem:myerson-phi} implies that the two quantities are equal. If we take expectation over $s$, $\mathbb{E}_{b,s}\left[(\tilde{\varphi}(b)-s)\cdot \ind[\tilde{\varphi}(b)-s>0]\right]\leq \som$ according to Lemma~\ref{lem:single_ppm}. Similarly, we have $\mathbb{E}_{b,s}\left[(b-\tilde{\tau}(s))\cdot \ind[b-\tilde{\tau}(s)>0]\right]\leq \bom$. Our claim follows from combining these two inequalities with Inequality~\eqref{equ:gft-separate}.
\end{proof}

\section{Double Auctions}\label{sec:single dimensional}
In Section~\ref{sec:single}, we proposed two simple mechanisms approximating the optimal GFT for the bilateral trading setting. Both of the mechanisms are IR, BIC and SBB. However in double auctions, such a mechanism appears to be hard to design directly. One significant barrier is to find a payment rule that simultaneously satisfies all three conditions mentioned above. Indeed, given an allocation rule, even monotone, such a payment rule is not guaranteed to exist. Myerson and Satterthwaite~\cite{MyersonS83} characterized the monotone allocation rules for which such payment rules exist in the bilateral trading setting. They also provided an explicit description of the payment. In Section~\ref{sec:characterization}, we provide a similar characterization for the double auction setting (see Theorem~\ref{thm:characterization}). However, even knowing there exists a payment rule that makes the mechanism IR, BIC and SBB, it is still not easy to explicitly describe these payments.

We circumvent this difficulty by first considering a weaker BB constraint -- ex-ante WBB, then apply Theorem~\ref{thm:ex-ante WBB to SBB} to transform the designed mechanism into an SBB one. In this section, we propose a simple, IR, DSIC and ex-ante WBB mechanism whose GFT is at least $\frac{\opt}{2}$.

\begin{theorem}\label{thm:ewbb mechanism}
In any double auction, there exists an IR, DSIC and ex-ante WBB mechanism whose GFT is at least $\frac{\opt}{2}$.
\end{theorem}

Again, we propose two mechanisms and show that the better one has GFT at least $\frac{1}{2}\opt$. The allocation and payment rule are defined as follows:

\begin{definition}\label{def:som and bom}
We consider the following two mechanisms which are the generalizations of SOM and BOM in double auctions.
\begin{itemize}
	\item \textbf{Generalized Seller-Offering Mechanism (GSOM)}: Given type profile $\Bb,\Bs$, find a matching $A(\Bb,\Bs)\in \mathcal{F}$ that maximizes $\sum_{(i,j)\in A(\Bb,\Bs)}(\tilde{\varphi}_i(b_i)-s_j)$. We use $A^{1}(\Bb,\Bs)$ to denote the maximum weight matching. For each $(i,j)\in A^{1}(\Bb,\Bs)$, buyer $i$ trades with seller $j$. According to Lemma~\ref{lem:matching}, this allocation rule is monotone, and the buyer (or the seller) pays (or receives) the \emph{threshold payment}. {See below for more details about the payment rule.}

	\item \textbf{Generalized Buyer-Offering Mechanism (GBOM)}: Given type profile $\Bb,\Bs$, find a matching $A(\Bb,\Bs)\in \mathcal{F}$ that maximizes $\sum_{(i,j)\in A(\Bb,\Bs)}(b_i-\tilde{\tau}_j(s_j))$. We use $A^{2}(\Bb,\Bs)$ to denote the maximum weight matching. For each $(i,j)\in A^{2}(\Bb,\Bs)$, buyer $i$ trades with seller $j$. According to Lemma~\ref{lem:matching}, this allocation rule is monotone, and the buyer (or the seller) pays (or receives) the \emph{threshold payment}.
\end{itemize}
\end{definition}

Given type profile $(\Bb,\Bs)$, both mechanisms first build a complete bipartite graph between the buyers and the sellers with edge weight $w_{ij}(b_i,s_j)$ equal to $\tilde{\varphi}_i(b_i)-s_j$ (or $b_i-\tilde{\tau}_j(s_j)$) for the edge between buyer $i$ and seller $j$. Then a maximum weight matching $A(\Bb,\Bs)$ within the family $\mathcal{F}$ is chosen.

\begin{lemma}\label{lem:matching}
Suppose $w_{ij}(b_i,s_j)$ is non-decreasing in $b_i$ and non-increasing in $s_j$ for every buyer $i$ and seller $j$. For any type profile $\Bb,\Bs$, if edge $(i,j)$ is in the maximum weight matching $\MM$, this edge is in the maximum weight matching under type profile $(b_i',b_{-i},\Bs)$ (or $(\Bb,s_j',s_{-j})$) for any $b_i'>b_i$ (or any $s_j'<s_j$).
\end{lemma}

\begin{proof}
We prove that for any $b_i'>b_i$, $\MM$ is still a maximum weight matching under type profile $(b_i',b_{-i},\Bs)$. For convenience, we use $w_{ij}(\Bb,\Bs)$ to represent the weight of edge $(i,j)$ under type profile $(\Bb,\Bs)$. $w_{ij}(\Bb,\Bs)=w_{ij}(b_i,s_j)$. For every matching $\MM'\in \mathcal{F}$, notice that $w_{i'j'}(\Bb,\Bs)=w_{i'j'}(b_i',b_{-i},\Bs)$ for all $i'\not=i$ and
all $j'\in[m]$. Hence,
\begin{equation}
\begin{aligned}
\sum_{(i',j')\in \MM}w_{i'j'}(b_i',b_{-i},\Bs)
&=\sum_{(i',j')\in \MM}w_{i'j'}(\Bb,\Bs)+\left(w_{ij}(b_i',s_j)-w_{ij}(b_i,s_j)\right)\\
&\geq \sum_{(i',j')\in \MM'}w_{i'j'}(\Bb,\Bs)+(w_{ij}(b_i',s_j)-w_{ij}(b_i,s_j))~~\text{(Optimality of $\MM$)}\\
&\geq \sum_{(i',j')\in \MM'}w_{i'j'}(b_i',b_{-i},s)
\end{aligned}
\end{equation}

The last inequality is an equality if $(i,j)\in \MM'$. If $(i,j)\notin \MM'$, the inequality is because $w_{ij}(b_i',s_j)-w_{ij}(b_i,s_j)\geq 0$ and
$\sum_{(i',j')\in \MM'}w_{i',j'}(\Bb,\Bs)=\sum_{(i',j')\in \MM'}w_{i'j'}(b_i',b_{-i},\Bs)$. Thus, $\MM$ is a maximum weight matching under type profile $(b_i',b_{-i},\Bs)$. Similarly, we can show that $\MM$ is  a maximum weight matching under type profile $(\Bb,s_j,s_{j'})$ for all $s_j'<s_j$.
\end{proof}

By Lemma~\ref{lem:matching}, the allocation rules for GSOM and GBOM are both monotone due to the monotonicity of functions $\tilde{\varphi}_i(\cdot)$ and $\tilde{\tau}_j(\cdot)$ for all $i,j$. We use the \emph{threshold payment}, that is, given any type profile $\Bb,\Bs$, for every buyer $i$, if $x_i^B(\Bb,\Bs)=0$, $p_i^B(\Bb,\Bs)$ is also $0$, and if $x_i^B(\Bb,\Bs)=1$, $p_i^B(\Bb,\Bs)$ equals the smallest $b_i'$ such that $x_i^B(b_i',b_{-i},\Bs)=1$. Similarly, for every seller $j$, if $x_j^S(\Bb,\Bs)=0$, $p_j^S(\Bb,\Bs)$ is also $0$, and if $x_j^S(\Bb,\Bs)=1$, $p_j^S(\Bb,\Bs)$ equals to the largest $s_j'$ such that $x_j^S(\Bb,s_j',s_{-j})=1$.
 As the allocation rule is monotone and we use {threshold payments}, the mechanism is IR and DSIC for every agent.

\subsection{GSOM and GBOM in Bilateral Trading}
As a warmup, we show that the GSOM (or the GBOM) is ex-ante WBB in the bilateral trading setting.
Let $T^B=\{\beta_1,...,\beta_{K}\}$ where $\beta_1<\beta_2<...<\beta_{K}$ and $T^S=\{\alpha_1,...,\alpha_{K'}\}$ where $\alpha_1<\alpha_2<...<\alpha_{K'}$. For GSOM, the pair is selected in the optimal matching if and only if the weight $\tilde{\varphi}(b)-s>0$. Notice that this is exactly the allocation rule used in SOM.
\begin{lemma}\label{lem:GSOM in bt}
In the Bilateral Trading setting, GSOM is an IR, DSIC and ex-ante WBB mechanism.
\end{lemma}
\begin{proof}
We only need to prove that GSOM is ex-ante WBB. By definition of the mechanism, if $x^B(b,s)=x^S(b,s)=1$, $p^B(b,s)=\beta_{k(s)}$ where $k(s)$ is the smallest index such that $\tilde{\varphi}\left(\beta_{k(s)}\right)\geq s$, which depends on $s$. On the other hand, $p^S(b,s)=\alpha_{c(b)}$ where $c(b)$ is the largest index such that $\alpha_{c(b)}\leq\tilde{\varphi}(b)$, which depends on $b$.

For every $s$,
$$\mathbb{E}_b\left[p^B(b,s)\right]=\beta_{k(s)}\cdot \sum_{r=k(s)}^{K} f^B(\beta_r),$$ and
$$\mathbb{E}_b\left[p^S(b,s)\right]=\sum_{r=k(s)}^{K}\alpha_{c(\beta_r)}\cdot f^B(\beta_r)
\leq\sum_{r=k(s)}^{K}\tilde{\varphi}(\beta_r)\cdot f^B(\beta_r).$$

By definition of $k(s)$, it is easy to see that $k(s)$ is not in the interior of any ironed interval, so $\beta_{k(s)}\cdot\sum_{r=k(s)}^{K} f^B(\beta_r)=\sum_{r=k(s)}^{K}\tilde{\varphi}(\beta_r)\cdot f^B(\beta_r)$ by Lemma~\ref{lem:myerson-phi}. Therefore, $\mathbb{E}_b\left[p^S(\Bb,\Bs)\right]\leq\mathbb{E}_b\left[p^B(\Bb,\Bs)\right]$. Taking expectation over $s$ on both sides finishes the proof.
\end{proof}

Similarly, we can prove that GBOM is ex-ante WBB. The trade happens if and only if the weight $b-\tilde{\tau}(s)>0$ in GBOM, which is exactly the same allocation rule with BOM. Using a similar argument as in Lemma~\ref{lem:GSOM in bt}, we can prove the following Lemma.
\begin{lemma}
In the Bilateral Trading setting, GBOM is an IR, DSIC and ex-ante WBB mechanism.
\end{lemma}

\subsection{The General Case}
In this section, we consider the general double auction setting, and we first argue that both GSOM and GBOM are ex-ante WBB. The idea is to consider each pair $(i,j)$ separately and show that the expected payment of buyer $i$ for trading with seller $j$ is greater than the expected gains of seller $j$ for trading with buyer $i$. The formal proof uses similar techniques as in Lemma~\ref{lem:GSOM in bt}.

\begin{lemma}\label{lem:ewbb proof}
Both GSOM and GBOM are IR, DSIC and ex-ante WBB mechanisms.
\end{lemma}
\begin{proof}
We will give the proof for GSOM and the same argument applies to GBOM. Let $(x^B,x^S,p^B,p^S)$ be the allocation and payment rule for GSOM. For every $i,j$ and type profile $\Bb,\Bs$, let $x_{ij}(\Bb,\Bs)=\ind[(i,j)\in A^1(\Bb,\Bs)]$ representing whether buyer $i$ is trading with seller $j$. Clearly $x_i^B(\Bb,\Bs)=\sum_{j
\in [m]} x_{ij}(\Bb,\Bs)$, $x_j^S(\Bb,\Bs)=\sum_{i\in[n]} x_{ij}(\Bb,\Bs)$. We notice that with threshold payments, $p^B_i(\Bb,\Bs)$ (or $p^S_j(\Bb,\Bs)$) is non-zero only if $x^B_i(\Bb,\Bs)$ (or $x^S_j(\Bb,\Bs)$) is 1. Then the difference between all buyers' expected payments and sellers' expected gains can be written as
\begin{equation}\label{equ:diff of payment}
\begin{aligned}
\mathbb{E}_{\Bb,\Bs}\bigg[\sum_{i=1}^n x^B_i(\Bb,\Bs)p^B_i(\Bb,\Bs)-&\sum_{j=1}^m x^S_j(\Bb,\Bs)p^S_j(\Bb,\Bs)\bigg]\\
=&\mathbb{E}_{\Bb,\Bs}\bigg[\sum_{(i,j)\in A^1(\Bb,\Bs)}x_{ij}(\Bb,\Bs)\cdot\left(p^B_i(\Bb,\Bs)-p^S_j(\Bb,\Bs)\right)\bigg]\\
=&\mathbb{E}_{\Bb,\Bs}\bigg[\sum_{i,j}x_{ij}(\Bb,\Bs)\cdot\left(p^B_i(\Bb,\Bs)-p^S_j(\Bb,\Bs)\right)\bigg]\\
=&\sum_{i=1}^n\sum_{j=1}^m\mathbb{E}_{\Bb,\Bs}\left[x_{ij}(\Bb,\Bs)\cdot\left(p^B_i(\Bb,\Bs)-p^S_j(\Bb,\Bs)\right)\right]
\end{aligned}
\end{equation}

Now we fix $i,j$, $b_{-i}$ and $s_{-j}$. Lemma~\ref{lem:matching} states that if a pair $(i,j)$ is in the max weight matching then increasing the value of $b_i$ or decreasing the value of $s_j$ will not remove this pair from the max weight matching. In other words, $~x_{ij}(b_i,b_{-i},s_j,s_{-j})$ is non-decreasing in $b_i$ and non-increasing in $s_j$. {Next, we characterize the threshold payments of $i$ and $j$.} Again, let $T^B_i=\{\beta_1,...,\beta_{K}\}$ where $\beta_1<\beta_2<...<\beta_{K}$ and $T^S_j=\{\alpha_1,...,\alpha_{K'}\}$ where $\alpha_1<\alpha_2<...<\alpha_{K'}$. For every $b_i\in T^B_i$, define $c(b_i)$ to be the largest index $c$ such that $x_{ij}(b_i,b_{-i},\alpha_c,s_{-j})=1$. Notice that when $(i,j)\in A^1(b_i,b_{-i},\alpha_c,s_{-j})$, $\tilde{\varphi}_i(b_i)-\alpha_c$ must be non-negative, as the feasibility constraint $\mathcal{F}$ is downward-closed, so removing a pair with negative weight gives a strictly better matching. Thus, $\alpha_{c(b_i)}\leq \tilde{\varphi}_i(b_i)$.

Similarly, for every $s_j\in T_j^S$, define $k(s_j)$ to be the smallest index $k$ such that $x_{ij}(\beta_k,b_{-i},s_j,s_{-j})=1$. By the definition of threshold payments, given $b_i,s_j$, if $x_{ij}(\Bb,\Bs)=1$, $p^S_j(\Bb,\Bs)=\alpha_{c(b_i)}$. {As for the buyer,} $p^B_i(\Bb,\Bs)=\beta_{k(s_j)}$. {The reason is that when $s_j\geq \alpha_{c(b_i)}$ (or $b_i\leq \beta_{k(s_j)}$) then $x_j^S(\Bb,\Bs)$ (or $x_i^B(\Bb,\Bs)$) must be $0$. Imagine this is not the case, and $j$ (or $i$) is in the maximum matching with some other buyer $i'$ (or seller $j'$) under profile $(\Bb,\Bs)$, then clearly if we decrease the value of $s_j$ (or increase the value of $b_i$), $(i',j)$ (or $(i,j')$) should remain in the maximum matching according to Lemma~\ref{lem:matching}. Contradiction.}

Now fix $s_j$, $x_{ij}(\Bb,\Bs)=1$ if and only if $b_i\geq \beta_{k(s_j)}$ by the definition of $k(s_j)$.
we have
\begin{equation}
\mathbb{E}_{b_i}\left[x_{ij}(\Bb,\Bs)p^B_i(\Bb,\Bs)\right]=\beta_{k(s_j)}\cdot \sum_{k=k(s_j)}^K f^B_i(\beta_k)
\end{equation}
\begin{equation}
\mathbb{E}_{b_i}\left[x_{ij}(\Bb,\Bs)p^S_i(\Bb,\Bs)\right]=\sum_{k=k(s_j)}^K f^B_i(\beta_k)\cdot \alpha_{c(\beta_k)}\leq \sum_{k=k(s_j)}^K f^B_i(\beta_k)\cdot \tilde{\varphi}_i(\beta_k)
\end{equation}

Notice again that $\beta_{k(s_j)}$ does not lie in the interior of any ironed interval for all $s_j$. So $\beta_{k(s_j)}\cdot$\\
\noindent$\sum_{k=k(s_j)}^K f^B_i(\beta_k)= \sum_{k=k(s_j)}^K f^B_i(\beta_k)\cdot \tilde{\varphi}_i(\beta_k)$ by Lemma~\ref{lem:myerson-phi}. Hence, $$\mathbb{E}_{b_i}\left[x_{ij}(\Bb,\Bs)p^B_i(\Bb,\Bs)\right]\geq \mathbb{E}_{b_i}\left[x_{ij}(\Bb,\Bs)p^S_i(\Bb,\Bs)\right].$$

Take expectation over $s_j, b_{-i}, s_{-j}$, and sum over all $i,j$:
\begin{equation}\label{inequ:diff of payment}
\sum_{i=1}^n\sum_{j=1}^m\mathbb{E}_{\Bb,\Bs}\left[x_{ij}(\Bb,\Bs)(p^B_i(\Bb,\Bs)-p^S_j(\Bb,\Bs))\right]\geq 0
\end{equation}

Hence, GSOM is ex-ante WBB.
\end{proof}

Now we are ready to prove Theorem~\ref{thm:ewbb mechanism}. 

\begin{prevproof}{Theorem}{thm:ewbb mechanism}
We use $\gsom$ (or $\gbom$) to denote the expected GFT of GSOM (or GBOM).  According to Lemma~\ref{lem:ewbb proof}, both GSOM and GBOM are IR, BIC and ex-ante WBB, so we only need to prove that $\gsom+\gbom\geq \opt$. By Equation~\ref{equ:lagrangian with phi} and Lemma~\ref{lem:canonial flow}, when $\alpha=1$, for any IR, BIC and SBB mechanism $M=(A,x^B,x^S,p^B,p^S)$,
\begin{equation}\label{equ:gft-separate-general}
\begin{aligned}
\gft(M)&\leq \mathbb{E}_{\Bb,\Bs}\left[\sum_{i=1}^n x^B_i(\Bb,\Bs)\cdot (b_i+\tilde{\varphi}_i(b_i))-\sum_{j=1}^m x^S_j(\Bb,\Bs)\cdot (s_j+\tilde{\tau}_j(s_j))\right]\\
&=\mathbb{E}_{\Bb,\Bs}\left[\sum_{(i,j)\in A(\Bb,\Bs)}(b_i+\tilde{\varphi}_i(b_i)-s_j-\tilde{\tau}_j(s_j))\right]~~~~~~\text{(Definition of $A(\Bb,\Bs)$)}\\
&\leq \mathbb{E}_{\Bb,\Bs}\left[\sum_{(i,j)\in A^{1}(\Bb,\Bs)}(\tilde{\varphi}_i(b_i)-s_j)\right]+\mathbb{E}_{\Bb,\Bs}\left[\sum_{(i,j)\in A^{2}(\Bb,\Bs)}(b_i-\tilde{\tau}_j(s_j))\right]~~~\text{(Definition~\ref{def:som and bom})}\\
\end{aligned}
\end{equation}

For every $i, j$, fix $b_{-i},\Bs$, we continue to use the notation $k(s_j)$ and $x_{ij}(\Bb,\Bs)$ as in Lemma~\ref{lem:ewbb proof} according to GSOM.
\begin{equation}
\begin{aligned}
\mathbb{E}_{b_i}[(\tilde{\varphi}(b_i)-s_j)&\cdot x_{ij}(\Bb,\Bs)]\\
=&\sum_{r=k(s_j)}^K\tilde{\varphi}(\beta_r)\cdot f^B(\beta_r)-s_j\cdot \mathbb{E}_{b_i}[x_{ij}(\Bb,\Bs)]\\
=&\beta_{k(s_j)}\cdot \sum_{r=k(s_j)}^K f^B(\beta_r)-s_j\cdot \sum_{r=k(s_j)}^K f^B(\beta_r) \quad\text{($\beta_{k(s_j)}$ not in interior of any ironed interval)}\\
\leq &\sum_{r=k(s_j)}^K (\beta_r-s_j)\cdot f^B(\beta_r)\\
=&\mathbb{E}_{b_i}\left[(b_i-s_j)\cdot x_{ij}(\Bb,\Bs)\right]
\end{aligned}
\end{equation}

Take expectation on $b_{-i},\Bs$, and then sum up over all $i,j$:
\begin{equation}
\begin{aligned}
\mathbb{E}_{\Bb,\Bs}\left[\sum_{(i,j)\in A^{1}(\Bb,\Bs)}(\tilde{\varphi}_i(b_i)-s_j)\right]&=\sum_{i,j}\mathbb{E}_{\Bb,\Bs}\left[(\tilde{\varphi}(b_i)-s_j)\cdot x_{ij}(\Bb,\Bs)\right]\\
&\leq \sum_{i,j}\mathbb{E}_{\Bb,\Bs}\left[(b_i-s_j)\cdot x_{ij}(\Bb,\Bs)\right]= \gsom
\end{aligned}
\end{equation}
Similarly, we have $\mathbb{E}_{\Bb,\Bs}\left[\sum_{(i,j)\in A^{2}(\Bb,\Bs)}\left(b_i-\tilde{\tau}_j(s_j)\right)\right]\leq \gbom$. Combine this with Inequality~(\ref{equ:gft-separate-general}), we have $\opt\leq \gsom+\gbom$.
\end{prevproof}

Finally, by combining Theorem~\ref{thm:ex-ante WBB to SBB} with Theorem~\ref{thm:ewbb mechanism}, we have constructed a simple, IR, BIC and SBB mechanism whose GFT is at least $\frac{\opt}{2}$.
\begin{theorem}
	In any double auction with arbitrary buyer and seller value distributions and arbitrary downward-closed feasibility constraint $\mathcal{F}$, we can design an IR, BIC, SBB mechanism whose GFT is at least $\frac{\opt}{2}$.\end{theorem}
\notshow{\begin{proof}
	It follows from Theorem~\ref{thm:ex-ante WBB to SBB} and Theorem~\ref{thm:ewbb mechanism}.
\end{proof}}

\section{Transformation to SBB Mechanisms}\label{sec:transformation}
In this section, we argue how to transform an IR, BIC and ex-ante WBB mechanism to an IR, BIC and SBB mechanism without changing the allocation rule. Clearly, the GFT remains the same after the transformation. We want to emphasize that all transformations in this section are not restricted to double auctions and can be applied to general two-sided markets.

We first clarify the definition of a mechanism $M=(x^B, x^S, p^B, p^S)$ used in Section~\ref{sec:transformation}, for a general two-sided market. For every type profile $(\Bb,\Bs)$, for every buyer $i\in [n]$ (or seller $j\in [m]$), $x^B_i(\Bb,\Bs)$ (or $x_j^S(\Bb,\Bs)$) describes the allocation rule for buyer $i$ (or seller $j$) under this type profile. For example, in the two sided-market with multiple heterogeneous items, $x^B_i(\Bb,\Bs)$ contains the probability for buyer $i$ to receive any subset of items. We use  $v^B_i(x^B_i(\Bb,\Bs))$ (or $v^S_j(x_j^S(\Bb,\Bs))$) to denote buyer $i$'s (or seller $j$'s) expected value of the allocation under the type profile $(\Bb,\Bs)$.

If a mechanism is ex-ante WBB, the sum of buyers' expected payment is larger than the sum of sellers' expected gains. While fixing the allocation rule, if we simply give the surplus to each seller evenly in the end of the mechanism, the utility of any agent does not decrease under any type profile, hence the mechanism remains to be IR. Also, since the surplus is distributed evenly and independent of the reported type, the mechanism is still BIC. The following lemma provides the transformation from an arbitrary IR, BIC, ex-ante WBB mechanism to an IR, BIC, ex-ante SBB mechanism.

\begin{lemma}\label{lem:ewbb to ebb}
Given an IR, BIC, ex-ante WBB mechanism $M=(x^B, x^S, p^B, p^S)$ with non-negative payment rule, there exists another non-negative payment rule $(p^{B'}, p^{S'})$ such that mechanism $M'=(x^B, x^S, p^{B'}, p^{S'})$ is IR, BIC and ex-ante SBB.
\end{lemma}
\begin{proof}
Let $\delta=\mathbb{E}_{\Bb,\Bs}\left[\sum_{i=1}^n p^B_i(\Bb,\Bs)-\sum_{j=1}^m p^S_j(\Bb,\Bs)\right]\geq 0$. Define $(p^{B'}, p^{S'})$ as follows: for every $\Bb,\Bs$, $p^{B'}(\Bb,\Bs)=p^B(\Bb,\Bs)\geq 0$, $p^{S'}(\Bb,\Bs)=p^S(\Bb,\Bs)+\frac{\delta}{m}\geq 0$. Then
\begin{equation}
\mathbb{E}_{\Bb,\Bs}\left[\sum_{i=1}^n p^{B'}_i(\Bb,\Bs)-\sum_{j=1}^m p^{S'}_j(\Bb,\Bs)\right]=\mathbb{E}_{\Bb,\Bs}\left[\sum_{i=1}^n p^B_i(\Bb,\Bs)-\sum_{j=1}^m p^S_j(\Bb,\Bs)\right]-m\cdot \frac{\delta}{m}=0
\end{equation}
$M'$ is ex-ante SBB. In mechanism $M'$, it first gives $\frac{\delta}{m}$ to each seller and then follows the allocation rule and payment of mechanism $M$. Since each seller receives a fixed amount of money at the beginning of the mechanism, $M'$ will still be IR and BIC.
\end{proof}

The next theorem provides a transformation for turning an IR, BIC, ex-ante SBB mechanism to a SBB mechanism without modifying the allocation rule. 

\begin{theorem}\label{thm:ebb to sbb}
Given an IR, BIC and ex-ante SBB mechanism $M=(x^B, x^S, p^B, p^S)$ with non-negative payment rule, there exists another non-negative payment rule $(p^{B'}, p^{S'})$ such that the mechanism $M'=(x^B, x^S, p^{B'}, p^{S'})$ is IR, BIC and SBB.
\end{theorem}

\begin{proof}
We will construct $p^{B'}$ such that for every buyer $i$ and her type $b_i$, the expected payment for buyer $i$ to report type $b_i$ in $M'$ is the same as her payment in $M$. Formally,
$$\mathbb{E}_{b_{-i},\Bs}\left[p^{B'}_i(b_i,b_{-i},\Bs)\right]=\mathbb{E}_{b_{-i},\Bs}\left[p^{B}_i(b_i,b_{-i},\Bs)\right].$$
Similarly, for each seller $j$ and her type $s_j$, the expected gains for seller $j$ to report type $s_j$ in $M'$ is the same as her gains in
$M$, that is, $$\mathbb{E}_{s_{-j},\Bb}\left[p^{S'}_j(\Bb,s_j,s_{-j})\right]=\mathbb{E}_{s_{-j},\Bb}\left[p^{S}_j(\Bb,s_j,s_{-j})\right].$$

 This property guarantees that the expected utility for buyer $i$ (or seller $j$) when reporting type $b_i$ (or type $s_j$) stays unchanged. Since $M$ is BIC and IR, $M'$ is also  BIC and IR. 

Suppose we are given an IR, BIC and ex-ante SBB mechanism $M=(x^B, x^S, p^B, p^S)$. Define the payment rule $(p^{B'}, p^{S'})$ as follows. Let $\Omega^B=\left\{i\in [n]: \mathbb{E}_{\Bb',\Bs'}[p^B_{i}(\Bb',\Bs')]>0\right\}$ and $\Omega_S=\{j\in [m] : \mathbb{E}_{\Bb',\Bs'}[p^S_{j}(\Bb',\Bs')]>0\}$.

For $i\not\in \Omega^B$, since all payments $p^B_{i}(\Bb',\Bs')$ are non-negative, we must have $p^B_{i}(\Bb',\Bs')=0$ for all $\Bb',\Bs'$. We define $p^{B'}_i(\Bb,\Bs)=0$ for every type profile $(\Bb,\Bs)$. Similarly, for $j\not\in \Omega^S$, let $p^{S'}_j(\Bb,\Bs)=0$ for every type profile $(\Bb,\Bs)$. We use $p_i^B(b_i)$ to denote $\mathbb{E}_{b_{-i},\Bs}\left[p^{B}_i(b_i,b_{-i},\Bs)\right]$ and $p_j^S(s_j)$ to denote $\mathbb{E}_{s_{-j},\Bb}\left[p^{S}_j(\Bb,s_j,s_{-j})\right]$.

For $i\in\Omega^B$ and $j\in \Omega^S$, define

\begin{equation}\label{equ:ebb to sbb-buyer payment}
p^{B'}_i(\Bb,\Bs)=\frac{\prod_{i'\in \Omega^B} p^B_{i'}(b_{i'})\cdot \prod_{j'\in\Omega^S} p^S_{j'}(s_{j'})}{\prod_{i'\in \Omega^B} \mathbb{E}_{\Bb',\Bs'}[p^B_{i'}(\Bb',\Bs')]\cdot \prod_{j'\in\Omega^S} \mathbb{E}_{\Bb',\Bs'}[p^S_{j'}(\Bb',\Bs')]}\cdot \mathbb{E}_{\Bb',\Bs'}[p_i^B(\Bb',\Bs')]
\end{equation}

\begin{equation}
p^{S'}_j(\Bb,\Bs)=\frac{\prod_{i'\in \Omega^B} p^B_{i'}(b_{i'})\cdot \prod_{j'\in\Omega^S} p^S_{j'}(s_{j'})}{\prod_{i'\in \Omega^B} \mathbb{E}_{\Bb',\Bs'}[p^B_{i'}(\Bb',\Bs')]\cdot \prod_{j'\in\Omega^S} \mathbb{E}_{\Bb',\Bs'}[p^S_{j'}(\Bb',\Bs')]}\cdot \mathbb{E}_{\Bb',\Bs'}[p_j^S(\Bb',\Bs')]
\end{equation}

For every $\Bb,\Bs$, since $M$ is ex-ante SBB, $\sum_{i\in \Omega^B} \mathbb{E}_{\Bb',\Bs'}[p^B_{i}(\Bb',\Bs')]=\sum_{j\in \Omega^S} \mathbb{E}_{\Bb',\Bs'}[p_j^S(\Bb',\Bs')]$, which implies that $\sum_{i\in \Omega^B} p^{B'}_i(\Bb,\Bs)=\sum_{j\in \Omega^S} p^{S'}_j(\Bb,\Bs)$. Since the payments of buyers (or sellers) that are not in $\Omega_B$ (or $\Omega_S$) are $0$, Mechanism $M'$ is SBB. Moreover, for every $i\in \Omega^B$ and type $b_i$, if we take expectation of $p^{B'}_i(\Bb,\Bs)$ over all $b_{-i},\Bs$, we have
\begin{equation}
\begin{aligned}
\mathbb{E}_{b_{-i},\Bs}[p^{B'}_i(\Bb,\Bs)]&=\frac{\displaystyle{p^B_i(b_i)\cdot \prod_{i'\not=i, i'\in\Omega^B} \mathbb{E}_{b_{i'}\sim D^B_{i'}}[p^B_{i'}(b_{i'})]\cdot \prod_{j'\in\Omega^S} \mathbb{E}_{s_{j'}\sim D^S_{j'}}[p^S_{j'}(s_{j'})]\cdot \mathbb{E}_{\Bb',\Bs'}[p_i^B(\Bb',\Bs')]}}{\displaystyle{\prod_{i'\in\Omega^B} \mathbb{E}_{\Bb',\Bs'}[p^B_{i'}(\Bb',\Bs')]\cdot \prod_{j'\in\Omega^S} \mathbb{E}_{\Bb',\Bs'}[p^S_{j'}(\Bb',\Bs')]}}\\
&=p^B_i(b_i)
\end{aligned}
\end{equation}

If $i\not\in \Omega^B$, $\mathbb{E}_{b_{-i},\Bs}[p^{B'}_i(\Bb,\Bs)]=0=p^B_i(b_i)$. Similarly, for every seller $j$ and any of her type $s_j$, $\mathbb{E}_{\Bb,s_{-j}}[p^{S'}_j(\Bb,\Bs)]=p^S_j(s_j)$. Thus, $M'$ is an IR, BIC and SBB mechanism.
\end{proof}

With Lemma~\ref{lem:ewbb to ebb} and Theorem~\ref{thm:ebb to sbb}, one can show that under IR and BIC constraints, mechanisms with all variants of the Budget-Balance constraint can be transformed to one another, without changing the allocation rule (and thus not affecting the GFT). Figure~\ref{fig:transformation} describes the transformation between mechanisms with different Budget-Balance constraints. In the figure, all the simple arrows are directed from a stronger constraint to a weaker one. 

\begin{figure}
\centering
\includegraphics[scale=0.5]{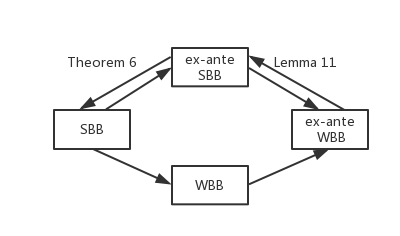}
\caption{Transformation between Mechanisms with Different BB Constraints}
\label{fig:transformation}
\end{figure}

\section{Characterizing the Implementable Allocation Rules in Double Auctions}\label{sec:characterization}

In this section, we characterize the set of allocation rules that are implementable by an IR, BIC, SBB mechanism in Theorem~\ref{thm:characterization}. It generalizes Myerson and Satterthwaite's result~\cite{Myerson81} to double auctions. In particular, their result is a special case of ours when $n=m=1$.

Up till now, all agents' distributions are assumed to have finite support. For simplicity, in Theorem~\ref{thm:characterization} we assume all agents have continuous distributions. In other words, for every buyer $i$ (or seller $j$), $f^B_i(\cdot)$ (of $f^S_j(\cdot)$) is a continuous function and positive in its domain $[\underline{b_i},\overline{b_i}]$ (or $[\underline{s_j},\overline{s_j}]$).

\begin{theorem}\label{thm:characterization}
Given an allocation rule $x=(x^B, x^S)$, there exists payment rule $(p^B, p^S)$ such that the mechanism $M=(x^B, x^S, p^B, p^S)$ is IR, BIC and SBB if and only if
\begin{itemize}
\item For every $i$, $x^B_i(b_i)$ is non-decreasing on $b_i$. For every $j$, $x^S_j(s_j)$ is non-increasing on $s_j$.
\item
\begin{equation}\label{equ:thm-characterization}
\mathbb{E}_{\Bb,\Bs}\left[\sum_{i=1}^n x^B_i(\Bb,\Bs)\cdot \varphi_i(b_i)-\sum_{j=1}^m x^S_j(\Bb,\Bs)\cdot \tau_j(s_j)\right]\geq 0
\end{equation}
\end{itemize}
\end{theorem}

\begin{prevproof}{Theorem}{thm:characterization}

{The following lemma characterizes the payments in an IR, BIC mechanism. The proof is similar to the analysis for bilateral trading in~\cite{MyersonS83}.}

\begin{lemma}\label{lem:charac}
Suppose a mechanism $M=(x^B,x^S,p^B,p^S)$ is IR and BIC, then
\begin{itemize}
\item For every $i$, $x^B_i(b_i)$ is non-decreasing on $b_i$. For every $j$, $x^S_j(s_j)$ is non-increasing on $s_j$.
\item For every buyer $i$ and any of her type $b_i$,
\begin{equation}\label{equ:threshold payment-buyer}
p_i^B(b_i)=b_i\cdot x_i^B(b_i)-\int_{\underline{b_i}}^{b_i}x_i^B(t)dt-\theta_i,
\end{equation}
$\theta_i$ is some non-negative constant.
\item For every seller $j$ and any of her type $s_j$,
\begin{equation}\label{equ:threshold payment-seller}
p_j^S(s_j)=s_j\cdot x_j^S(s_j)+\int_{s_j}^{\overline{s_j}}x_j^S(t)dt+\eta_j,
\end{equation}
$\eta_j$ is some non-negative constant.
\end{itemize}
Furthermore, if $(p^B,p^S)$ satisfies Equation~\eqref{equ:threshold payment-buyer} and \eqref{equ:threshold payment-seller}, then
\begin{equation}\label{equ:threshold payment-diff}
\begin{aligned}
\sum_{i=1}^n\mathbb{E}_{b_i}[p_i^B(b_i)]-\sum_{j=1}^m\mathbb{E}_{s_j}[p_j^S(s_j)]&=\mathbb{E}_{\Bb,\Bs}\left[\sum_{i=1}^n x^B_i(\Bb,\Bs)\cdot \varphi_i(b_i)-\sum_{j=1}^m x^S_j(\Bb,\Bs)\cdot \tau_j(s_j)\right]\\
&-\left(\sum_{i=1}^n \theta_i+\sum_{j=1}^m \eta_j\right)
\end{aligned}
\end{equation}
\end{lemma}

\begin{proof}
For every $i$ and $b_i$, let $U^B_i(b_i)$ be buyer $i$'s expected utility when she reports her true type $b_i$. Since $M$ is BIC, for every buyer $i$, and two types $b_i,b_i'\in [\underline{b_i},\overline{b_i}]$,
\[U^B_i(b_i)=b_i\cdot x_i^B(b_i)-p_i^B(b_i)\geq b_i\cdot x_i^B(b_i')-p_i^B(b_i')\]
\[U^B_i(b_i')=b_i'\cdot x_i^B(b_i')-p_i^B(b_i')\geq b_i'\cdot x_i^B(b_i)-p_i^B(b_i)\]
The two inequalities imply
\[(b_i'-b_i)\cdot x_i^B(b_i)\leq U^B_i(b_i')-U^B_i(b_i)\leq (b_i'-b_i)\cdot x_i^B(b_i')\]

when $b_i'>b_i$, $x_i^B(b_i')\geq x_i^B(b_i)$. $x_i^B(b_i)$ is non-decreasing on $b_i$ and thus Riemann integrable. Let $b_i'=b_i+\epsilon$,
\begin{equation}
\epsilon\cdot x_i^B(b_i)\leq U^B_i(b_i+\epsilon)-U^B_i(b_i)\leq \epsilon\cdot x_i^B(b_i+\epsilon)
\end{equation}

For any value $z$, taking integral of $b_i$ on $[\underline{b_i},z]$ and let $\epsilon\to 0$, we have
\begin{equation}
U^B_i(z)=U^B_i(\underline{b_i})+\int_{\underline{b_i}}^z x_i^B(b_i)db_i
\end{equation}
\begin{equation}
p^B_i(z)=z\cdot x_i^B(z)-\int_{\underline{b_i}}^z x_i^B(b_i)db_i-U^B_i(\underline{b_i})
\end{equation}

Clearly, $U^B_i(\underline{b_i})\geq 0$, as $M$ is interim IR. Equation~\eqref{equ:threshold payment-buyer} follows from setting $\theta_i=U^B_i(\underline{b_i})$. Similarly, we can show that Equation~\eqref{equ:threshold payment-seller} holds for every $j$.

Furthermore, for every $i$, by Equation~\eqref{equ:threshold payment-buyer},
\begin{equation}
\begin{aligned}
\mathbb{E}_{b_i}[p_i^B(b_i)]&=\int_{\underline{b_i}}^{\overline{b_i}}b_i x_i^B(b_i)f^B_i(b_i)db_i-\int_{\underline{b_i}}^{\overline{b_i}}\int_{\underline{b_i}}^{b_i}x_i^B(t)f^B_i(b_i)dtdb_i-\theta_i\\
&=\int_{\underline{b_i}}^{\overline{b_i}}b_i x_i^B(b_i)f^B_i(b_i)db_i-\int_{\underline{b_i}}^{\overline{b_i}}x_i^B(t)\int_{t}^{\overline{b_i}}f^B_i(b_i)db_i dt-\theta_i\\
&=\int_{\underline{b_i}}^{\overline{b_i}}b_i x_i^B(b_i)f^B_i(b_i)db_i-\int_{\underline{b_i}}^{\overline{b_i}}x_i^B(b_i) \left(1-F_i^B(b_i)\right)db_i-\theta_i\\
&=\int_{\underline{b_i}}^{\overline{b_i}}\varphi_i(b_i)x_i^B(b_i)f^B_i(b_i)db_i-\theta_i
\end{aligned}
\end{equation}
Similarly, for every seller $j$, by Equation~\eqref{equ:threshold payment-seller},
\begin{equation}
\begin{aligned}
\mathbb{E}_{s_j}[p_j^S(s_j)]&=\int_{\underline{s_j}}^{\overline{s_j}}s_j x_j^S(s_j)f^S_j(s_j)ds_j+\int_{\underline{s_j}}^{\overline{s_j}}\int_{s_j}^{\overline{s_j}}x_j^S(t)f^S_j(s_j)dtds_j+\eta_j\\
&=\int_{\underline{s_j}}^{\overline{s_j}}s_j x_j^S(s_j)f^S_j(s_j)ds_j\int_{\underline{s_j}}^{\overline{s_j}}x_j^S(t)\int_{\underline{s_j}}^{t}f^S_j(s_j)ds_j dt+\eta_j\\
&=\int_{\underline{s_j}}^{\overline{s_j}}s_j x_j^S(s_j)f^S_j(s_j)ds_j\int_{\underline{s_j}}^{\overline{s_j}}x_j^S(s_j)F_j^S(s_j) ds_j+\eta_j\\
&=\int_{\underline{s_j}}^{\overline{s_j}}\tau_j(s_j)x_j^S(s_j)f^S_j(s_j)ds_j+\theta_j
\end{aligned}
\end{equation}

Equation~\eqref{equ:threshold payment-diff} directly follows from the two equations above.

\end{proof}

\hspace{0.1in}

If there exists a payment rule $(p^B, p^S)$ such that the mechanism $M=(x^B, x^S, p^B, p^S)$ is IR, BIC and SBB, by Lemma~\ref{lem:charac}, $x$ is monotone and there exists non-negative $\theta_i$'s and $\eta_j$'s such that Equation~\eqref{equ:threshold payment-diff} holds. Since $M$ is SBB,
\begin{equation}
\mathbb{E}_{\Bb,\Bs}\left[\sum_{i=1}^n x^B_i(\Bb,\Bs)\cdot \varphi_i(b_i)-\sum_{j=1}^m x^S_j(\Bb,\Bs)\cdot \tau_j(s_j)\right]=\sum_{i=1}^n \theta_i+\sum_{j=1}^m \eta_j\geq 0
\end{equation}

If the given allocation rule $x$ is monotone, and satisfies Inequality~\eqref{equ:thm-characterization}, define payment rule $p=(p^B, p^S)$ as follows:
\begin{equation}
p_i^B(\Bb,\Bs)=b_i\cdot x_i^B(\Bb,\Bs)-\int_{\underline{b_i}}^{b_i}x_i^B(t,b_{-i},\Bs)dt
\end{equation}
\begin{equation}
p_j^S(\Bb,\Bs)=s_j\cdot x_j^S(\Bb,\Bs)+\int_{s_j}^{\overline{s_j}}x_j^S(\Bb,t,s_{-j})dt
\end{equation}

This is a threshold payment and thus $M=(x,p)$ is DSIC and IR. According to Equation~\eqref{equ:threshold payment-diff} and Inequality~\eqref{equ:thm-characterization}, $M$ is ex-ante WBB. By Lemma~\ref{lem:ewbb to ebb} and Theorem~\ref{thm:ebb to sbb}, there exists payment rule $(p^{B'}, p^{S'})$ such that $M'=(x^B, x^S, p^{B'}, p^{S'})$ is IR, BIC, and SBB.
\end{prevproof}

\bibliographystyle{plain}	
\bibliography{Yang}
\newpage
\appendix
\section{Missing Proofs from Section~\ref{sec:canonical flow}}\label{appx:ironing}
\begin{prevproof}{Lemma}{lem:canonial flow}

For any buyer $i$, if $D^B_i$ is irregular, i.e., $\varphi_i(\cdot)$ is not monotonically non-decreasing, we iron the virtual value function $\Phi_i^B(b_i)$ as in~\cite{CaiDW16}: for every $b_i,b_i'$ such that $b_i<b_i'$ but $\Phi_i^B(b_i)-b_i>\Phi_i^B(b_i')-b_i'$, we add a loop between $b_i$ and $b_i'$ with proper weight $w$. Then by Definition~\ref{def:virtual value}, $\Phi_i^B(b_i)$ decreases by $\frac{w\cdot(b_i'-b_i)}{f^B_i(b_i)}$ while $\Phi_i^B(b_i')$ increases by $\frac{w\cdot(b_i'-b_i)}{f^B_i(b_i')}$. By choosing an appropriate weight $w$, we can make $\Phi_i^B(b_i)-b_i=\Phi_i^B(b_i')-b_i'$. If we keep adding loops, we can make sure $\Phi_i^B(b_i)-b_i$ is monotone, and it equals to the Myerson's ironed virtual value function $\tilde{\varphi}_i(b_i)$ as defined in~\cite{CaiDW16}.

For any seller $j$, if $D^S_j$ is irregular, i.e., $\tau_j(\cdot)$ is not monotonically non-decreasing, we iron the virtual value function $\Phi_j^S(s_j)$ in the following way: for every $s_j,s_j'$ such that $s_j>s_j'$ but $\Phi_j^S(s_j)-s_j<\Phi_j^S(s_j')-s_j'$, we add a loop between $s_j$ and $s_j'$ with proper weight $w$. Then by Definition~\ref{def:virtual value}, $\Phi_j^S(s_j)$ increases by $\frac{w\cdot(s_j-s_j')}{f^S_j(s_j)}$ while $\Phi_j^S(s_j')$ decreases by $\frac{w\cdot(s_j-s_j')}{f^S_j(s'_j)}$. If we choose $w$ appropriately, we can make $\Phi_j^S(s_j)-s_j=\Phi_j^S(s_j')-s_j'$. If we keep adding loops, we can make sure $\Phi_j^S(s_j)-s_j$ is monotone, and $\Phi_j^S(s_j)-s_j=\alpha\cdot\tilde{\tau}_j(s_j)$, where $\tilde{\tau}_j(s_j)$ is the Myerson's ironed virtual value function for the seller defined as follows.

Let $F(\cdot)$ be cdf for $D^S_j$. Define $U(q)=q\cdot F^{-1}(q)$ for any $q\in[0,1]$. It is not hard to show that when $\tau_j(s_j)$ is not monotone $U(q)$ is not convex. Let $\{\alpha_1,...,\alpha_{K}\}$ be the support set of $D^S_j$, $\alpha_1<\alpha_2<...<\alpha_{K}$. For every $k\in [K]$, define $\widetilde{U}(F(\alpha_k))$ as follows:
\begin{equation}
\widetilde{U}(F(\alpha_k))=\min_{k_1,k_2\in [K]}\big(\delta\cdot U(F(\alpha_{k_1}))+(1-\delta)\cdot U(F(\alpha_{k_2})\big)
\end{equation}
where $\delta\in [0,1]$ is the unique value such that $F(\alpha_{k})=\delta\cdot F(\alpha_{k_1})+(1-\delta)\cdot F(\alpha_{k_2})$.

With the definition above, the set of points $(F(s_j),\widetilde{U}(F(s_j)))$ forms a convex curve and stays below the original points $(F(s_j),U(F(s_j)))$. The ironed virtual value is then defined as
\begin{equation}
\tilde{\tau}_j(\alpha_k)=\frac{\widetilde{U}(F(\alpha_{k}))-\widetilde{U}(F(\alpha_{k-1}))}{F(\alpha_{k})-F(\alpha_{k-1})}
\end{equation}
for $2\leq k\leq K$ and $\tilde{\tau}_j(\alpha_1)=\tau_j(\alpha_1)=\alpha_1$. Then $\tilde{\tau}_j(\cdot)$ is a monotonically non-decreasing function.

\end{prevproof}

\section{Missing Proofs from Section~\ref{sec:single}}\label{sec:proof single}
\begin{prevproof}{Lemma}{lem:myerson-phi}
If $D^B$ is regular, $\tilde{\varphi}(b)=\varphi(b)$ for all $b$. By the definition of Myerson's virtual value in Lemma~\ref{lem:flow without ironing},
\begin{equation}
\begin{aligned}
\sum_{r=k}^{K}\varphi(\beta_r)\cdot f^B(\beta_r)&=\sum_{r=k}^{K-1} \left(\beta_rf^B(\beta_r)-(\beta_{r+1}-\beta_r)\cdot \sum_{l=r+1}^{K} f^B(\beta_l)\right)+\beta_Kf^B(\beta_{K})\\
&=\sum_{r=k}^{K-1} \left(\beta_r\cdot \sum_{l=r}^Kf^B(\beta_l)-\beta_{r+1}\cdot \sum_{l=r+1}^{K} f^B(\beta_l)\right)+\beta_Kf^B(\beta_K)\\
&=\beta_k\cdot \sum_{r=k}^{K} f^B(\beta_r)
\end{aligned}
\end{equation}

If $D^B$ is not regular, consider the revenue curve in the quantile space. For each $\beta_k$, $\sum_{r=k}^{K}\varphi(\beta_r)\cdot f^B(\beta_r)$ is the expected revenue for selling the item at price $\beta_k$ which corresponds to the value of the revenue curve at $1-F_i^B(\beta_k)$. By the definition of the ironed virtual value, $\sum_{r=k}^{K}\tilde{\varphi}(\beta_r)\cdot f^B(\beta_r)$ corresponds to the value of the ironed revenue curve at $1-F_i^B(\beta_k)$. Since the ironed revenue curve never goes below the revenue curve,
\begin{equation}
\sum_{r=k}^{K}\tilde{\varphi}(\beta_r)\cdot f^B(\beta_r)\geq \sum_{r=k}^{K}\varphi(\beta_r)\cdot f^B(\beta_r)=\beta_k\cdot \sum_{r=k}^{K} f^B(\beta_r)
\end{equation}

If type $\beta_k$ does not lie in the interior of any ironed interval, the two curves meet at $\beta_k$. Hence, the equality sign holds.
\end{prevproof}

\begin{prevproof}{Lemma}{lem:myerson-tau}
If $D^S$ is regular, $\tilde{\tau}(s)=\tau(s)$ for all $s$. By the definition of Myerson's virtual value in Lemma~\ref{lem:flow without ironing},
\begin{equation}
\begin{aligned}
\sum_{r=1}^{k}\tau(\alpha_r)\cdot f^S(\alpha_r)&=\sum_{r=2}^{k} \left(\alpha_rf^S(\alpha_r)+(\alpha_{r}-\alpha_{r-1})\cdot \sum_{l=1}^{r-1} f^S(\alpha_l)\right)+\alpha_1f^S(\alpha_1)\\
&=\sum_{r=2}^{k} \left(\alpha_r\cdot \sum_{l=1}^rf^S(\alpha_l)-\alpha_{r-1}\cdot \sum_{l=1}^{r-1} f^S(\alpha_l)\right)+\alpha_1f^S(\alpha_1)\\
&=\alpha_k\cdot \sum_{r=1}^{k} f^S(\alpha_r)
\end{aligned}
\end{equation}

If $D^S$ is not regular, notice that in the quantile space, the ironed revenue curve is convex and never goes above the original revenue curve. Similarly, we have
\begin{equation}
\sum_{r=1}^{k}\tilde{\tau}(\alpha_r)\cdot f^S(\alpha_r)\leq \sum_{r=1}^{k}\tau(\alpha_r)\cdot f^S(\alpha_r)=\alpha_k\cdot \sum_{r=1}^{k} f^S(\alpha_r)
\end{equation}

If type $\alpha_k$ does not lie in the interior of any ironed interval, the two curves meet at $\alpha_k$. Hence, the equality sign holds.
\end{prevproof}



\end{document}
